\documentclass[12pt, draftclsnofoot, onecolumn]{IEEEtran}
\usepackage{amsfonts}
\usepackage{slashbox}
\usepackage{amsmath}
\usepackage{multirow}
\usepackage{graphicx}
\usepackage{amsfonts}
\usepackage{amsmath,epsfig}
\usepackage{mathbbold}
\usepackage{stfloats}
\usepackage{amssymb}
\usepackage{url}
\usepackage{bm}
\usepackage{cite}
\usepackage{amsmath}
\usepackage{amssymb}
\usepackage{latexsym}
\usepackage{multirow}
\usepackage{epsfig}
\usepackage{graphics}
\usepackage{mathrsfs}
\usepackage{algorithmic}
\usepackage{algorithm}
\usepackage{subfigure}
\usepackage{slashbox}
\usepackage[all]{xy}

\usepackage{pifont}
\usepackage{bbding}
\usepackage{stmaryrd}
\usepackage{amssymb}
\usepackage{amsfonts}
\usepackage{epic}
\usepackage{stfloats}
\usepackage{latexsym}
\usepackage{epstopdf}
\usepackage{epic}
\usepackage{bm}
\usepackage{xcolor}
\usepackage{mathrsfs}
\usepackage{pifont}
\usepackage{bbding}
\usepackage{amsmath,epsfig}
\usepackage{mathbbold}
\usepackage{stmaryrd}
\usepackage{amssymb}
\usepackage{amsfonts}
\usepackage{epic}
\usepackage{graphicx}
\usepackage{subfigure}

\usepackage{enumerate}

\usepackage{stfloats}
\usepackage{latexsym}
\usepackage{epstopdf}
\usepackage{epic}
\usepackage{multirow}
\usepackage{stfloats}
\usepackage{bm}
\usepackage{amsmath}
\usepackage{amssymb}
\usepackage{amsthm}

\usepackage{booktabs}
\usepackage{color}
\usepackage{cases}
\usepackage{array}
\usepackage{makecell}
\usepackage{times}

\newtheorem{proposition}{Proposition}

\newcommand{\Ss}{\bm S}
\newcommand{\G}{\bm G}
\newcommand{\HH}{\bm H}
\newcommand{\R}{\bm R}
\newcommand{\V}{\bm V}

\newcommand{\I}{\bm{I}}

\newcommand{\W}{\bm W}
\newcommand{\PP}{\bm P}

\newcommand{\Mm}{\bm{M}}
\newcommand{\x}{ \bm{x} }

\newcommand{\ttheta}{\mathbf \Theta}

\newcommand{\aaa}{\bm{a}}

\newcommand{\vv}{\bm v}

\newcommand{\g}{\bm g}
\newcommand{\h}{\bm h}

\graphicspath{{./figs/}}

\begin{document}
\title{Intelligent Reflecting  Surface Enhanced Wireless Network via Joint Active and
Passive Beamforming}

\author{\IEEEauthorblockN{Qingqing Wu,  \emph{Member, IEEE} and Rui Zhang, \emph{Fellow, IEEE}
\thanks{ The authors are with the Department of Electrical and Computer Engineering, National University of Singapore, email:\{elewuqq, elezhang\}@nus.edu.sg. Part of this work has been presented in  \cite{wu2018IRS}.}   }  }

\maketitle
\textheight 9in
 \voffset -0.1in
\vspace{-0.6in}
\begin{abstract}
Intelligent reflecting surface (IRS) is a  revolutionary and transformative   technology for achieving spectrum  and energy efficient wireless communication  cost-effectively in the future.
Specifically, an IRS consists of a large number of low-cost passive elements each being able to reflect  the incident signal independently with an adjustable phase shift so as to collaboratively achieve three-dimensional (3D) passive beamforming without the need of any transmit radio-frequency (RF) chains.
{In this paper, we study an IRS-aided single-cell  wireless system where one IRS is deployed to assist in the communications between  a multi-antenna access point (AP) and multiple single-antenna users.}
We formulate and solve new problems  to minimize the total transmit power at the AP by jointly optimizing the transmit beamforming by active antenna array at the AP and reflect beamforming by passive phase shifters at the IRS, subject to users' individual signal-to-interference-plus-noise ratio (SINR) constraints.  
    {Moreover, we analyze the asymptotic performance of  IRS's passive beamforming with infinitely large  number of reflecting elements and compare it to that of  the traditional active beamforming/relaying.
 Simulation results demonstrate the significant performance gain achieved by the proposed scheme with IRS over a benchmark massive MIMO system  without using IRS. We also draw useful insights into optimally deploying IRS in future wireless systems.}
\end{abstract}
\vspace{-0.5cm}
\vspace{-0.5cm}
\section{Introduction}
  To achieve 1,000-fold network capacity increase and ubiquitous wireless connectivity for at least 100 billion devices  in the forthcoming fifth-generation (5G)  networks, a variety of wireless technologies have been proposed and thoroughly investigated in the last decade, including most prominently the ultra-dense network (UDN), massive multiple-input multiple-output (MIMO), and  millimeter wave (mmWave) communication \cite{boccardi2014five}. However, the network energy consumption and hardware cost still remain critical issues  in practical systems  \cite{zhang2016fundamental}. For example, UDNs almost linearly scale up the circuit and cooling energy consumption with the number of  deployed  base stations (BSs), while costly radio frequency (RF) chains  and complex signal processing are needed for achieving high-performance  communication at mmWave frequencies, especially when massive MIMO is employed to exploit the small wavelengths.  Moreover, adding an excessively large number of  active components such as small-cell BSs/relays/remote radio heads (RRHs)  in wireless networks also causes a more aggravated  interference issue. As such,  innovative research on finding both spectrum and energy efficient techniques with low hardware cost is still imperative for realizing a sustainable wireless network evolution with scalable cost in the future  \cite{wu2016overview}.

In this paper, intelligent reflecting surface (IRS) is proposed as a promising new  solution to achieve the above goal \cite{JR:wu2019IRSmaga,JR:wu2019discreteIRS,JR:wu2019SWIPT:IRS,JR:xinrong:IRS}.  
Specifically, IRS is a planar array consisting of  a large number of reconfigurable passive elements (e.g., low-cost printed dipoles), where each of the elements is able to induce a certain phase shift  (controlled by an attached smart controller) independently on the incident signal, thus collaboratively changing the reflected signal propagation. Although passive reflecting surfaces have found a variety of applications in radar systems, remote sensing, and satellite/deep-space communications, they were rarely used in mobile wireless communication. This is mainly because traditional reflecting surfaces only have fixed phase shifters once  fabricated, which are unable  to cater to the dynamic wireless channels arising from user mobility. However, recent advances in RF micro electromechanical systems (MEMS) and metamaterial  (e.g., metasurface)  have made the reconfigurability of reflecting surfaces possible, even by controlling the phase shifters  in real time \cite{cui2014coding}. By smartly adjusting the phase shifts of all passive  elements at the IRS, the reflected signals  can add coherently with the signals from other paths  at the desired receiver to boost the received signal power or destructively at non-intended receivers to suppress  interference as well as  enhancing security/privacy \cite{JR:wu2019IRSmaga,JR:xinrong:IRS}.

\begin{table*}[!t]
{{\caption{Comparison of IRS with other related  technologies.}\label{table1}
\vspace{-0.3cm}
\tiny
\centering
 \linespread{1.2}
\small
\begin{tabular}{|m{2.25cm}|m{2.3cm}|m{1.6cm}|m{2.0cm}|m{1.4cm}|m{2.0cm}|m{1.7cm}|}
  \hline
{\!Technology} &{Operating ~~~~~mechanism}&  {Duplex} &  No.  of transmit RF   chains  needed &{Hardware cost} &{Energy} \textcolor{white}{xx} {consumption} & {Role} \\ \hline
{IRS}                 &Passive, ~~~~~~~~~~~~~ reflect &  {Full duplex}  & $0$ & Low &  Low &   Helper     \\ \hline
{Backscatter}&  Passive, ~~~~~~~~~~~~ reflect &  {Full duplex}  & $0$ &Very low& Very low   &  Source    \\ \hline
{MIMO  relay}                 & Active, ~~~~~~~~~~~~~~~receive and\textcolor{white}{xxxx} transmit  & {Half/full duplex} & $N$ &High & High  &   Helper   \\ \hline
{Massive MIMO}&  Active, ~~~~~~~~~~transmit/receive&  {Half/full duplex} & $N$ &Very high& Very high  & {Source/   Destination}    \\ \hline
\end{tabular}  }\vspace{-0.7cm} }
\end{table*}

{It is worth noting that  the proposed IRS differs significantly from other related existing technologies such as amplify-and-forward (AF) relay, backscatter communication, and active intelligent surface based massive MIMO \cite{JR:wu2019IRSmaga}.
First,  compared to the AF relay that assists in source-destination transmission by amplifying and regenerating signals, IRS does not use a transmitter module but only reflects the received  signals  as a passive array, which thus incurs no transmit power consumption.\footnote{{Although using devices like MEMs as mentioned previously to adjust the phase shifts at the IRS requires  some power consumption, it is practically negligible as compared to the much higher transmit power of active communication devices.}} Furthermore, active AF relay usually operates in half-duplex (HD) mode and thus is less spectrally efficient than the proposed IRS operating in full-duplex (FD) mode. Although AF relay can also work in FD,  it inevitably suffers from the severe self-interference, which needs effective  interference cancellation techniques.
Second, different from the  traditional backscatter communication of the radio frequency identification (RFID)  tag that communicates with the receiver by reflecting the signal  sent from a reader, IRS is utilized mainly  to enhance the existing communication link performance instead of delivering its own information by reflection. As such,  the direct-path signal (from reader to receiver) in backscatter communication is  undesired interference and hence needs to be canceled/suppressed at the receiver.  However, in IRS-enhanced communication, both the direct-path and reflect-path signals carry the same useful information and thus can be   coherently added at the receiver to maximize the total received power. Third, IRS is also different from the active intelligent surface based massive MIMO \cite{hu2017beyond} due to their different array architectures (passive versus active) and operating mechanisms (reflect versus transmit).  {A more detailed comparison between the above technologies and IRS  is summarized in Table \ref{table1},  where $N$ denotes the number of  active antennas in massive MIMO or MIMO relay.} }

On the other hand, from the implementation  perspective, IRSs possess appealing  advantages such as low profile, lightweight, and conformal geometry, which enable them to be easily attached/removed to/from the wall or ceiling, thus providing high flexibility for their practical deployment  \cite{subrt2012intelligent}.
For example, by installing  IRSs on the walls/ceilings which are in line-of-sight (LoS) with an access point (AP)/BS,  the signal strength  in the vicinity of each IRS can be significantly improved. In addition, integrating IRSs into the existing networks (such as cellular or WiFi)  can be made transparent to the users without the need of any change in  the hardware and software of their  devices.
All the above features make  IRS a compelling new technology for future wireless networks, particularly in indoor applications with high density of users (such as  stadium, shopping mall, exhibition center, airport, etc.).
 To validate the feasibility of IRS, an experimental testbed for a two-user setup was developed \cite{tan2016increasing}, where the spectral efficiency is shown to be greatly improved by using the IRS. However, the research on IRS design as well as the performance analysis and optimization for IRS-aided wireless communication  systems is still in its infancy, which thus motivates this work.

{In this paper, we consider an IRS-aided  multiuser multiple-input single-output (MISO) communication  system in a single cell  as shown in Fig. \ref{system:model},  where a multi-antenna AP serves multiple single-antenna users with the help of an IRS.}   Since each user in general  receives the superposed (desired as well as interference) signals from both the AP-user (direct) link and AP-IRS-user (reflected) link,  we jointly optimize the (active) transmit beamforming at the AP and (passive) reflect beamforming by the phase shifters at the IRS to minimize the total transmit power at the AP, under a given set of signal-to-interference-plus-noise ratio (SINR) constraints at the user receivers.  For the special  case of single-user transmission without any interference,  it is intuitive that  the AP should  beam toward the user directly if the channel of the AP-user link is much stronger than that of the AP-IRS link; while in the opposite case, especially when the AP-user link is severally blocked by obstacles (e.g., thick walls  in indoor applications), the AP ought to  adjust its beam toward the IRS to maximally leverage its reflected signal to serve the user (i.e., by creating  a virtual LoS  link with the user to bypass  the obstacle).  In this case, a large number of reflecting elements with adjustable phases  at the IRS can focus the signal  into a sharp beam toward the user, thus achieving a high beamforming gain similarly as by the conventional massive MIMO \cite{Hien2013}, but only via a passive array with significantly reduced energy consumption and hardware cost.

Moreover, under the general multiuser setup, an IRS-aided  system will
  \begin{figure}[!t]
\centering
~~~~~~~~~\includegraphics[width=0.85\textwidth]{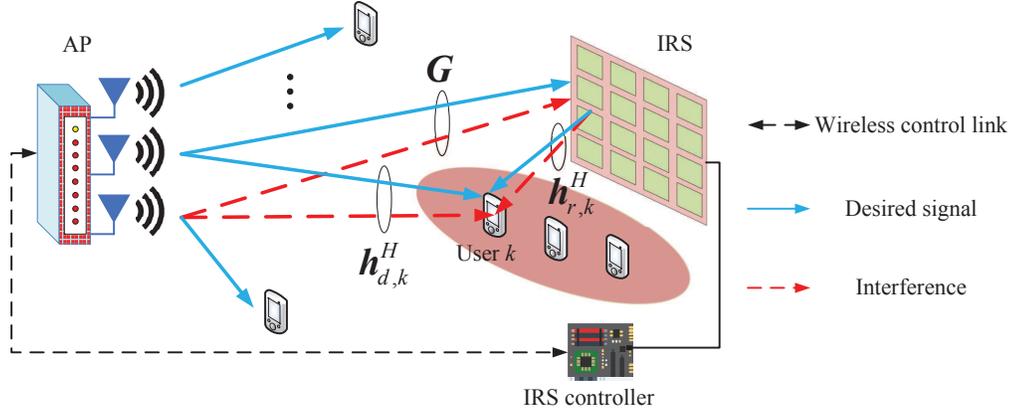}\vspace{-0.6cm}
\caption{An IRS-aided  multiuser  communication  system.    } \label{system:model} \vspace{-0.8cm}
\end{figure}
 benefit from two main aspects: the beamforming  of desired signal as in the single-user case as well as the spatial interference suppression among  the  users. Specifically, a user near the IRS is expected to be able to tolerate more  interference from the AP as compared to the user farther away from the IRS,  because the phase shifts of  the IRS can be  tuned such that the interference reflected by the IRS can add destructively with that from the AP-user link at the near user to suppress  its overall received  interference. This thus provides more flexibility  for designing the transmit beamforming at the AP for serving  the other users outside the IRS's covered region, so as  to improve the SINR performance of all users in the system.  Therefore, the transmit beamforming  at the AP needs to be jointly designed with the phase shifts at the IRS based on all the AP-IRS, IRS-users, and AP-users channels  in order to fully reap the network beamforming  gain.  However, this  design  problem is difficult to be solved optimally in general,  due to the non-convex SINR constraints as well as the signal unit-modulus constraints imposed by passive phase shifters.  Although beamforming optimization under unit-modulus constraints has been studied in the research on  constant-envelope precoding \cite{mohammed2012single,zhang2018constant} as well as  hybrid digital/analog processing  \cite{el2014spatially,foad16jstsp}, such designs are mainly restricted to either the transmitter  or the receiver side, which are not applicable to  our considered joint active and passive beamforming optimization at both the  AP and IRS.

To tackle this  new problem, we first consider a single-user setup and apply the semidefinite relaxation (SDR) technique to obtain a high-quality approximate solution as well as  a lower bound of the optimal value to evaluate the tightness of approximate  solutions. To reduce the computational complexity, we further propose an efficient algorithm based on the alternating optimization of the phase shifts and transmit beamforming vector in an iterative manner, where their optimal solutions are derived in closed-form with the other being fixed. Then, we extend our designs for the single-user case to the general multiuser setting, and  propose two algorithms to obtain suboptimal solutions that also offer  different tradeoffs between performance and complexity.
Numerical results demonstrate that the required transmit power at the AP to meet users' SINR targets can be considerably  reduced by deploying the IRS as compared to the conventional setup without using IRS for both single-user and multiuser setups.
In particular,  for serving a single-user in the vicinity of the IRS, it is shown  that the AP's transmit power decreases with the number of reflecting elements $N$ at the IRS in the order of  $N^2$ when $N$ is sufficiently large, which is consistent with the performance scaling law derived analytically. {Note that in \cite{huangachievable}, the authors also considered the use of passive intelligent mirror (analogous to IRS) to enhance the sum-rate in a multiuser system. This paper differs from \cite{huangachievable} in the following two main aspects.
 First, to simplify the system model and algorithm design, \cite{huangachievable} ignored the direct channels from the AP to users, while this paper considers the more general setting with the AP-user direct channels considered.  Second, \cite{huangachievable} adopted the suboptimal zero-forcing (ZF) based precoding at the AP to simplify the optimization of passive phase shifters, while in this paper we optimize AP transmit precoding jointly with IRS's phase shifts. As such, the algorithm proposed in \cite{huangachievable} is not applicable to solving the formulated problems in this paper. }

The rest of this paper is organized as follows. Section II introduces the system model and the problem formulation for designing the IRS-aided  wireless network.
In Sections III and IV, we propose efficient algorithms to solve the formulated problems in the single-user and multiuser cases, respectively.
Section V presents numerical results to evaluate  the performance of the proposed designs. Finally, we conclude the paper  in Section VI.

\emph{Notations:} Scalars are denoted by italic letters, vectors and matrices are denoted by bold-face lower-case and upper-case letters, respectively. $\mathbb{C}^{x\times y}$ denotes the space of $x\times y$ complex-valued matrices. For a complex-valued vector $\bm{x}$, $\|\bm{x}\|$ denotes its Euclidean norm, $\arg(\bm{x})$ denotes a vector with each element being the phase of the corresponding element in $\bm{x}$, and $\text{diag}(\bm{x})$ denotes a diagonal matrix with each diagonal element being the  corresponding element in $\bm{x}$.
The distribution of a circularly symmetric complex Gaussian (CSCG) random vector with mean vector  $\bm{x}$ and covariance matrix ${\bm \Sigma}$ is denoted by  $\mathcal{CN}(\bm{x},{\bm \Sigma})$; and $\sim$ stands for ``distributed as''. For a square matrix $\Ss$, ${\rm{tr}}(\Ss)$ and $\Ss^{-1}$ denote its trace and inverse, respectively, while $\Ss\succeq \bm{0}$ means that $\Ss$ is positive semi-definite.  For any general matrix $\Mm$, $\Mm^H$,  ${\rm{rank}}(\Mm)$, and $\Mm_{i,j}$ denote its conjugate transpose, rank, and $(i,j)$th element, respectively. $\I$ and $\bm{0}$ denote an identity matrix and an all-zero matrix, respectively, with appropriate dimensions. $\mathbb{E}(\cdot)$ denotes the statistical expectation. $ \mathrm{Re}\{\cdot\}$ denotes the real part of a complex number.

\section{System Model and Problem Formulation}
\subsection{System Model}
{As shown in Fig. \ref{system:model}, we consider the IRS-aided downlink communications in a single-cell network  where an IRS is deployed to assist in the communications from a multi-antenna AP to $K$ single-antenna users over a given frequency band.} The set of the users is denoted by  $\mathcal{K}$. The number of transmit antennas at the AP and that of reflecting units at the IRS are denoted by $M$ and $N$, respectively.  The IRS is equipped with a  controller that coordinates its switching between two working modes, i.e., receiving mode for  channel estimation and reflecting mode for data transmission \cite{JR:wu2019IRSmaga,subrt2012intelligent}. Due to the high path loss, it is assumed that the power of the signals that are reflected by the IRS two or more times  is negligible and thus ignored \cite{JR:wu2019IRSmaga}. To characterize the theoretical performance gain brought by the IRS, we assume that the channel state information (CSI) of all channels involved  is perfectly known at the AP. In addition, the quasi-static flat-fading  model is adopted for all channels.
  Since the IRS is a passive reflecting device,  we consider a time-division duplexing (TDD) protocol for  uplink and downlink transmissions and assume  channel reciprocity for the CSI acquisition in the downlink based on the uplink training.

%


{The baseband equivalent channels  from the AP to IRS, from the IRS to user $k$, and from the AP to user $k$ are denoted by $\bm{G}\in \mathbb{C}^{N\times M}$, $\bm{h}^H_{r,k}\in \mathbb{C}^{1\times N}$, and $\bm{h}^H_{d,k}\in \mathbb{C}^{1\times M}$, respectively, with  $k = 1, \cdots,K$.}  { It is worth noting that the reflected channel from the AP to each user via the IRS  is usually referred to as a dyadic backscatter channel in RFID communications \cite{JR:wu2019IRSmaga}, which behaves different from the AP-user direct channel. Specifically, each element of the IRS  receives the superposed multi-path signals from the transmitter, and then scatters the combined signal with adjustable amplitude and/or phase as if from a single point source \cite{JR:wu2019IRSmaga}.} Let $\bm{\theta}= [\theta_1, \cdots, \theta_N]$  and define a diagonal matrix  $\ttheta = \text{diag} (\beta_1 e^{j\theta_1}, \cdots, \beta_N e^{j\theta_N})$ (with $j$ denoting the imaginary unit) as the reflection-coefficients  matrix of the IRS, where $\theta_n\in [0, 2\pi)$ and $\beta_n \in [0, 1]$ denote the phase shift\footnote{To characterize the fundamental performance limits of IRS, we assume that the phase shifts can be continuously varied  in $[0, 2\pi)$, while in practice they are usually selected from a finite number of discrete values from 0 to $2\pi$ for the ease of circuit implementation. The design of IRS with discrete phase shifts is left for our future work \cite{JR:wu2019discreteIRS,wu2018IRS_discrete}.  }  and the amplitude reflection coefficient\footnote{In practice, each element of the IRS is usually designed to maximize the signal reflection. Thus, we set $\beta_n=1, \forall n$, in the sequel of this paper for simplicity. Note that this scenario is different from the traditional  backscatter communication where the RFID tags usually need to harvest a certain amount of energy from the incident signals for powering their  circuit operation and thus a much smaller amplitude reflection coefficient than unity  is resulted in general.} of the $n$th element of the IRS, respectively \cite{JR:wu2019IRSmaga}.
The composite AP-IRS-user channel  is thus modeled as a concatenation of three components, namely, the AP-IRS link, IRS reflection  with phase shifts, and IRS-user link. 

 In this paper, we consider linear transmit precoding at the AP where each user is assigned with one dedicated beamforming vector.  Hence, the complex baseband transmitted signal at the AP can be expressed as
$\bm{x}= \sum_{k=1}^K\bm{w}_ks_k$,
where $s_k$ denotes the transmitted data  for user $k$ and $\bm{w}_k\in \mathbb{C}^{M\times 1}$ is the corresponding  beamforming vector. It is assumed that $s_k$, $k = 1, \cdots,K$, are independent  random variables with zero mean and unit variance (normalized power).
The signal received at user $k$ from both the AP-user and AP-IRS-user channels is then expressed as
\begin{align}
y_k= ( \bm{h}^H_{r,k}\ttheta \bm{G} +  \bm{h}^H_{d,k}) \sum_{j=1}^K\bm{w}_js_j + n_k,  k=1, \cdots, K,
\end{align}
where $n_k \sim \mathcal{CN}(0, \sigma^2_k)$ denotes the  additive white Gaussian noise (AWGN) at the user $k$'s receiver.
Accordingly,  the SINR of user $k$ is given by
\begin{align}\label{eq:SINR}
\text{SINR}_k = \frac{|( \bm{h}^H_{r,k}\ttheta \bm{G}+\bm{h}^H_{d,k})\bm{w}_k |^2}{\sum_{j\neq k}^{K}|( \bm{h}^H_{r,k}\ttheta \bm{G}+\bm{h}^H_{d,k})\bm{w}_j |^2 +  \sigma^2_k}, k=1, \cdots, K.
\end{align}


\vspace{-0.3cm}
\subsection{Problem Formulation}
Let $\W = [{\bm w}_1, \cdots,{\bm w}_K]\in \mathbb{C}^{M\times K}$, $\HH_r = [{\bm h}_{r,1}, \cdots,{\bm h}_{r,K}]\in \mathbb{C}^{N\times K}$, and $\HH_d = [{\bm h}_{d,1}, \cdots,{\bm h}_{d,K}]\in \mathbb{C}^{M\times K}$.
In this paper, we aim to minimize the total transmit power at the AP by jointly optimizing the transmit beamforming at the AP and reflect beamforming at the IRS,  subject to individual SINR constraints at all users. Accordingly, the problem is formulated as
\begin{align}
\text{(P1)}: ~~\min_{\W, \bm{\theta}} ~~~&\sum_{k=1}^{K}\|\bm{w}_k\|^2 \\
\mathrm{s.t.}~~~~&\frac{|( \bm{h}^H_{r,k}\ttheta \bm{G}+\bm{h}^H_{d,k})\bm{w}_k |^2}{\sum_{j\neq k}^{K}|( \bm{h}^H_{r,k}\ttheta \bm{G}+\bm{h}^H_{d,k})\bm{w}_j |^2 +  \sigma^2_k}\geq \gamma_k, \forall k, \label{eq:coupling}\\
& 0\leq \theta_n \leq 2\pi, n=1,\cdots, N, \label{eq:modulus}
\end{align}
where $\gamma_k>0$ is  the minimum SINR requirement of user $k$.
Although the objective function of (P1) and constraints in \eqref{eq:modulus} are convex, it is challenging to solve (P1) due to the non-convex constraints in \eqref{eq:coupling} where the transmit beamforming and phase shifts are coupled. In general, there is no standard method for solving such non-convex optimization problems optimally.  Nevertheless, in the next section, we apply the SDR and alternating optimization techniques, respectively, to solve (P1) approximately  for the single-user case, which are then generalized to the multiuser case.
{Prior to solving problem (P1), we present a sufficient condition for its feasibility as follows. Let $\HH = [{\bm h}_1, \cdots,{\bm h}_K]\in \mathbb{C}^{M\times K}$ where ${\bm h}^H_k = \bm{h}^H_{r,k}\ttheta \bm{G}+\bm{h}^H_{d,k}$, $\forall k$.
\begin{proposition}\label{feasibility:condition}
Problem (P1) is feasible for any finite user SINR targets $\gamma_k$'s if ${\rm{rank}}(\G^H\HH_r + \HH_d)=K$.
\end{proposition}
\begin{proof}
If ${\rm{rank}}(\G^H\HH_r + \HH_d)=K$, the (right) pseudo inverse of $\HH^H =\HH_{r}^{H}\ttheta\G + \HH_{d}^{H}$ exists with $\ttheta = \I$ and the precoding matrix $\W$ at the AP can be set as
\begin{align}
\W =  \HH( \HH^H \HH)^{-1} \text{diag}(\gamma_1\sigma_1^2,\cdots, \gamma_k\sigma_k^2)^{\frac{1}{2}}.
\end{align}
It is easy to verify that the above solution allows all users to achieve their corresponding $\gamma_k$'s and  thus (P1) is feasible.
\end{proof} }

{
Thanks to the additional AP-IRS-user link, the rank condition in Proposition \ref{feasibility:condition} is practically easier to be satisfied  in an IRS-aided system, as
compared to that in the case without the IRS, i.e., ${\rm{rank}}(\HH_d)=K$. For instance, if the AP-user direct channels of two users lie in the same direction, then ${\rm{rank}}(\HH_d)=K$ does not hold. While the rank condition in an IRS-aided system may still hold  since the combined AP-user channels (including both the AP-user direct and AP-IRS-user reflected links) of these two users are unlikely to be aligned too, due to the additional IRS reflected paths.}
\vspace{-0.4cm}
\section{Single-User System}\label{sec:III}
In this section, we consider the single-user setup, i.e., $K=1$, to draw important  insights into the optimal joint beamforming design.  In this case, no inter-user interference is present, and thus (P1) is simplified to  (by dropping the user index)
\begin{align}
\text{(P2)}: ~~\min_{\bm{w}, \bm{\theta}} ~~~& \|\bm{w}\|^2  \label{eq:obj}\\
\mathrm{s.t.}~~~~&| (\bm{h}^H_r\ttheta \bm{G}+\bm{h}^H_d )\bm{w}|^2   \geq \gamma \sigma^2,   \label{SINR:constraints} \\
& 0\leq \theta_n \leq 2\pi,  n=1,\cdots, N. \label{phase:constraints}
\end{align}
Although much simplified, problem (P2) is still a non-convex optimization problem since the left-hand-side (LHS) of \eqref{SINR:constraints} is not jointly concave with respect to  $\bm{w}$ and $\bm{\theta}$.
 In the next two subsections, we  solve (P2) by applying the SDR and alternating optimization techniques, respectively, which will be extended to the general  multiuser system in the next section.
 \vspace{-0.4cm}
\subsection{SDR}

We first apply SDR to solve problem (P2), which also helps obtain a lower bound of the optimal value of (P2) for evaluating the performance gaps from other suboptimal solutions.  For any given phase shift $\bm{\theta}$, it is known that the maximum-ratio transmission (MRT) is the optimal transmit beamforming solution  to problem (P2) \cite{tse2005fundamentals}, i.e.,
$\bm w^* = \sqrt{ P} \frac{(\bm{h}^H_r\Theta \bm{G}+\bm{h}^H_d  )^H}{\|\bm{h}^H_r\ttheta \bm{G} +\bm{h}^H_d \|}$, where $P$ denotes the transmit power of the AP. Substituting $\bm w^*$ to  problem (P2) yields the following  problem
\begin{align}\label{secIII:p2}
 \min_{P, \bm{\theta}} ~~~&p \\
\mathrm{s.t.}~~~~&p\|\bm{h}^H_r\ttheta \bm{G}+ \bm{h}^H_d\|^2\geq \gamma\sigma^2,   \label{SecIII:SNRconstraint0}\\
&0\leq \theta_n \leq 2\pi, \forall n.  \label{SecIII:phaseconstraint0}
\end{align}
It is not difficult to verify that the optimal transmit power satisfies $P^* =  \frac{\gamma\sigma^2}{\|\bm{h}^H_r\ttheta \bm{G}+ \bm{h}^H_d\|^2}$. As such, minimizing the transmit power is equivalent to maximizing the channel power gain of the combined channel, i.e.,
\begin{align}\label{secIII:p3}
\max_{\bm{\theta}} ~~~&\|\bm{h}^H_r\ttheta \bm{G}+ \bm{h}^H_d\|^2\\
\mathrm{s.t.}~~~~&0\leq \theta_n \leq 2\pi, \forall n.  \label{SecIII:phaseconstraint}
\end{align}
Let $\bm{v} = [v_1, \cdots, v_N]^H$ where $v_n = e^{j\theta_n}$, $\forall n$. Then, constraints in \eqref{SecIII:phaseconstraint} are equivalent to  the unit-modulus constraints:  $|v_n|^2=1, \forall n$. By applying the change of variables $\bm{h}^H_r\ttheta \bm{G} =\bm{v}^H\bm{\Phi} $ where $\bm{\Phi}=\text{diag}(\bm{h}^H_r)\bm{G} \in \mathbb{C}^{N \times M}$, we have $\|\bm{h}^H_r\ttheta \bm{G} + \bm{h}^H_d\|^2 =\|\bm{v}^H\bm{\Phi}+ \bm{h}^H_d\|^2 $. Thus, problem \eqref{secIII:p3} is equivalent to
\begin{align}\label{secIII:p4}
\max_{\bm{v}} ~~~&\bm{v}^H\bm{\Phi}\bm{\Phi}^H\bm{v} + \bm{v}^H\bm{\Phi}\bm{h}_d+\bm{h}^H_d\bm{\Phi}^H \bm{v}  + \|\bm{h}^H_d\|^2\\
\mathrm{s.t.}~~~~& |v_n|^2=1, \forall n. \label{P4:C9}
\end{align}
Note that problem \eqref{secIII:p4} is a non-convex quadratically constrained quadratic program (QCQP), which can be reformulated as a homogeneous QCQP \cite{so2007approximating}. Specifically,
by introducing an auxiliary variable $t$, problem  \eqref{secIII:p4} is  equivalently written as
\begin{align} \label{secIII:p5}
\max_{\bm{\bar{v}}} ~~~&\bm{\bar{v}}^H\bm{R}\bm{\bar{v}}   + \|\bm{h}^H_d\|^2 \\
\mathrm{s.t.}~~~~& |\bar{v}_n|^2=1,  n=1,\cdots, N+1, \label{P4:C9}
\end{align}
where
\[
\bm{R}=\begin{bmatrix}
\bm{\Phi}\bm{\Phi}^H  & \bm{\Phi}\bm{h}_d \\
\bm{h}^H_d\bm{\Phi}^H  & 0 \\
\end{bmatrix},~~
\bm{\bar{v}}=\begin{bmatrix}
\bm{v}  \\
t \\
\end{bmatrix}.
\]
However, problem \eqref{secIII:p5} is still non-convex  in general \cite{so2007approximating}. 
Note that $\bm{\bar{v}}^H\bm{R}\bm{\bar{v}}={\rm{tr}}(\bm{R}\bm{\bar{v}}\bm{\bar{v}}^H)  $.
Define $\bm{V}=\bm{\bar{v}}\bm{\bar{v}}^H$, which needs to satisfy  $\bm{V}\succeq \bm{0}$ and ${\rm{rank}}(\bm{V})=1$. Since the rank-one constraint is non-convex, we apply SDR to relax this constraint. As a result,
problem \eqref{secIII:p5} is reduced to
\begin{align}\label{secIII:p6}
\max_{\bm{V}} ~~~&{\rm{tr}}(\bm{RV})   + \|\bm{h}^H_d\|^2 \\
\mathrm{s.t.}~~~~& \bm{V}_{n,n} = 1,  n=1,\cdots, N+1, \label{P6:C9} \\
~~~~&\bm{V} \succeq 0.  \label{P6:C9}
\end{align}
As problem \eqref{secIII:p6} is a convex semidefinite program (SDP), it can be optimally solved by existing convex optimization solvers such as CVX \cite{cvx}.
Generally,  the relaxed problem  \eqref{secIII:p6} may not lead to a rank-one solution, i.e., ${\rm{rank}}(\bm{V})\neq1$, which implies that  the optimal objective value of problem  \eqref{secIII:p6} only serves an upper bound of problem  \eqref{secIII:p5}. Thus, additional steps are needed to construct  a rank-one solution from the obtained higher-rank solution to problem  \eqref{secIII:p6}, while  the details can be found in \cite{wu2018IRS} and thus are omitted here.
 It has been shown that such an SDR approach followed by a sufficiently large number of randomizations  guarantees at least a $\frac{\pi}{4}$-approximation of the optimal objective value of problem \eqref{secIII:p6} \cite{so2007approximating}.

\vspace{-0.5cm}
\subsection{Alternating Optimization}
To achieve lower complexity than the SDR-based solution presented in the preceding subsection, we propose an alternative suboptimal  algorithm in this subsection based on alternating optimization.  Specifically, the transmit beamforming direction and transmit power at the AP  are optimized iteratively with the phase shifts at the IRS in an alternating manner, until the convergence is achieved.

Let $\bm{w}=\sqrt{P}\bm{\bar w}$ where $\bm{\bar w}$ denotes the transmit beamforming direction and $P$ is the transmit power. {For fixed transmit beamforming direction $\bm{\bar w}$,  (P2) is reduced to a joint transmit power and phase shifts optimization problem which can be formulated as (similar to \eqref{secIII:p2} and  \eqref{secIII:p3}),}
\begin{align}\label{secIII2:p3}
\max_{\bm{\theta}} ~~~&|(\bm{h}^H_r\ttheta \bm{G}+ \bm{h}^H_d){\bm{\bar w}}|^2\\
\mathrm{s.t.}~~~~&0\leq \theta_n \leq 2\pi, n=1,\cdots, N.  \label{SecIII2:phaseconstraint}
\end{align}
Although being non-convex, the above problem admits a closed-form solution by exploiting the special structure of its objective function. Specifically,  we have the following inequality:
\begin{align}\label{SecIV:distributed}
| (\bm{h}^H_r\ttheta \bm{G}+\bm{h}^H_d )\bm{\bar w}| &=| \bm{h}^H_r\ttheta \bm{G}\bm{\bar w}+ \bm{h}^H_d\bm{\bar w}|  \overset{(a)} \leq  | \bm{h}^H_r\ttheta \bm{G}\bm{\bar w}| + | \bm{h}^H_d\bm{\bar w}|,
\end{align}
where $(a)$ is due to the triangle inequality and the equality holds if and only if  $\mathrm{arg}(\bm{h}^H_r\ttheta \bm{G}\bm{\bar w})=\mathrm{arg}(   \bm{h}^H_d\bm{\bar w} )\triangleq \varphi_0$.
Next, we show that there always exists a solution $\bm{\theta}$ that satisfies $(a)$ with equality as well as  the phase shift constraints in \eqref{SecIII2:phaseconstraint}.
Let  $\bm{h}^H_r\ttheta\bm{G}\bm{w} =\bm{v}^H\bm{a}$ where $\bm{v} = [e^{j\theta_1}, \cdots, e^{j\theta_N}]^H$ and $\bm{a}=\text{diag}(\bm{h}^H_r)\bm{G}\bm{\bar w}$. With \eqref{SecIV:distributed},  problem \eqref{secIII2:p3} is equivalent to
\begin{align}
\max_{\bm{v}} ~~~&|\bm{v}^H\bm{a}|^2\\
\mathrm{s.t.}~~~~&  |v_n|=1, \forall n=1,\cdots, N,\\
& \mathrm{arg}(\bm{v}^H\bm{a})=  \varphi_0. 
\end{align}
It is not difficult to show that the optimal solution to the above problem is given by $\bm{v}^* = e^{j (\varphi_0 -  \arg(\aaa) )}=e^{j ( \varphi_0 - \arg( \text{diag}(\bm{h}^H_r)\bm{G}\bm{\bar w}) )}$. Thus, the $n$th phase shift at the IRS is given by
\begin{align}\label{phase:sub}
\theta^*_n &=\varphi_0 -  \arg({h}^H_{n,r}\bm{g}^H_{n}\bm{\bar w})= \varphi_0 -  \arg({h}^H_{n,r})- \arg(\bm{g}^H_{n}\bm{\bar w}),
\end{align}
where ${h}^H_{n,r}$ is the $n$th element of $\bm{h}^H_r$ and $\bm{g}^H_{n}$ is the $n$th row vector of $\bm{G}$. Note that $\bm{g}^H_{n}\bm{\bar w}$ combines the transmit beamforming and the AP-IRS channel, which can be regarded as the effective channel perceived by the $n$th reflecting element at the IRS.  Therefore, \eqref{phase:sub} suggests that the $n$th phase shift should be tuned such that the phase of the signal that passes through the AP-IRS and IRS-user links is aligned with that of the signal over the  AP-user direct  link to achieve coherent signal combining  at the user. Furthermore, it is interesting to note that the obtained phase $\theta^*_n $ is independent of the amplitude of ${h}_{n,r}$. As a result, the optimal  transmit power is given by $P^*=\frac{\gamma \sigma^2}{\|(\bm{h}^H_r\ttheta \bm{G}+\bm{h}^H_d){\bm{\bar w}} \|^2}$ from (P2).  Next, we optimize the transmit beamforming direction for given $\bm{\theta}$ in \eqref{phase:sub}.  As in Section III-A, the combined AP-user channel is given by $\bm{h}^H_r\ttheta \bm{G} +  \bm{h}^H_d$ and hence MRT is optimal, i.e., $\bm{\bar w}^{*} = \frac{(\bm{h}^H_r\ttheta \bm{G}+\bm{h}^H_d  )^H}{\|\bm{h}^H_r\ttheta \bm{G} +\bm{h}^H_d \|}$.  {The above alternating optimization approach is practically appealing since both the transmit beamforming and phase shifts are obtained in closed-form expressions, without invoking the SDP solver. Its convergence is guaranteed by the following two facts. First, for each subproblem, the optimal solution is obtained which ensures that the objective value of (P2) is non-increasing over iterations. Second, the optimal value of (P2) is bounded from below  due to the SNR constraint. Thus, the proposed algorithm  is guaranteed to converge. }

\vspace{-0.4cm}
{ \subsection{Power Scaling Law with Infinitely Large Surface}
{Next, we characterize the scaling law of the average received power at the user with respect to the number of reflecting elements, $N$,  in an IRS-aided system with $N \rightarrow \infty$. For simplicity, we assume $M=1$ with $\G\equiv \bm{g}$ to obtain essential insight. By ignoring the AP-user direct channel,  the user's received  power is  given by $P_u= P|{h}^H|^2=P |\bm{h}^H_r\ttheta \bm{g}|^2$. We consider three different phase shift solutions, i.e., 1) unit phase shift where $\ttheta =\I$; 2) random phase shift where $\theta_n$'s in $\ttheta$ are uniformly and randomly distributed in $[0, 2\pi)$; and 3) optimal phase shift which is obtained by the above two  proposed algorithms  (both optimal for $M=1$).
\begin{proposition}\label{scaling:law}
Assume $\bm{h}^H_{r} \sim \mathcal{CN}(\bm{0},\varrho^2_h\I)$ and $\bm{g} \sim \mathcal{CN}(\bm{0},\varrho^2_g\I)$.  As $N\rightarrow \infty$, it holds that
\begin{equation}\label{eq:scaling:law}
P_u\rightarrow
\left\{
\begin{aligned}
&NP\varrho^2_h\varrho^2_g,  && \text{for} ~~ \ttheta =\I~{\text{or  random}}~ \ttheta,  \\
&N^2\frac{P\pi^2\varrho^2_h\varrho^2_g}{16},  && \text{for}  ~~\text{optimal}~\ttheta.
\end{aligned}
\right.
\end{equation}
\end{proposition}  }
\begin{proof}
The three cases are discussed as follows:
\begin{itemize}
                \item When $ \ttheta =\I$, we have ${h}^H =  \bm{h}^H_{r} \bm{g}$. By invoking the Lindeberg-L$\acute{\rm e}$vy central limit theorem \cite{cramer2004random}, we have $ \bm{h}^H_{r} \bm{g} \sim \mathcal{CN}(0, N\varrho^2_h\varrho^2_g)$ as  $N\rightarrow \infty$. As the equivalent channel ${h}^H$ is a random variable, the average  user received power is given by $P_u= P\mathbb{E}(|{h}^H|^2) \rightarrow  NP\varrho^2_h\varrho^2_g$.
                    For the case of random phase shifts with $\theta_n\in [0, 2\pi)$, we have  ${h}^H =  \bm{h}^H_{r}\bm{\bar g}$ where $\bm{\bar g} =\ttheta \bm{g}$.   As  $\ttheta$ is a unitary matrix,    it follows  that $\bm{\bar g}$ has the same distribution as  $\bm{g}$, i.e., $\bm{\bar g} \sim \mathcal{CN}(\bm{0},\varrho^2_g\I)$. Thus we attain the same result as $ \ttheta =\I$.
                \item For optimal $\ttheta$ where the solution is given by \eqref{phase:sub}, we have  $|{h}^H |=|  \bm{h}^H_{r} \ttheta\bm{g} |=\sum_{n=1}^N |h_{r,n}||g_n|$, where $h_{r,n}$ and $g_n$ are the $n$-th elements in $\bm{h}^H_{r}$ and $\bm{g}$, respectively. The received power is then given by
                                        \begin{align} \label{IRS:receive:power}
P_u & =P \big|\sum_{n=1}^N |h_{r,n}||g_n| \big|^2.
\end{align}
                    Since $|h_{r,n}|$ and $|g_n|$ are statistically independent and follow Rayleigh distribution with mean values $\sqrt{\pi}\varrho_h/2$ and  $\sqrt{\pi}\varrho_g/2$, respectively, we have $\mathbb{E}(|h_{r,n}||g_n|)=\pi\varrho_h\varrho_g/4$. By using the fact that ${\sum_{n=1}^N|h_{r,n}||g_n|}/{N}\rightarrow \pi\varrho_h\varrho_g/4$ as  $N\rightarrow \infty$, it follows that
\begin{align}
P_u \rightarrow N^2\frac{P\pi^2\varrho^2_h\varrho^2_g}{16}.
\end{align}
\end{itemize}
This thus completes the proof.
\end{proof}   }

{The power scaling law with the optimal IRS phase design  in Proposition \ref{scaling:law} is highly promising  since it implies that by using a large number of reflecting units at the IRS, we can scale down the transmit power of the AP by a factor of $1/N^2$ without compromising the user received SNR. The fundamental reason behind such a ``squared gain'' is that the IRS not only achieves the transmit beamforming gain of order $N$ in the IRS-user link as in the conventional  massive MIMO \cite{Hien2013}, but also captures an inherent aperture gain of order $N$ by collecting more signal power in the AP-IRS link, which, however, cannot be achieved by scaling up the number of transmit antennas in massive MIMO due to the fixed total transmit power. Moreover, for the two benchmark cases with unit and random phase shifts at the IRS, a received power gain of order $N$ is also achieved. This shows the practical usefulness of the IRS, even without requiring any channel knowledge for  optimally setting the phase shifts. Note that  the received  noise power in the IRS-aided system remains constant as $N$ increases and thus the corresponding user receive   SNR also has the same squared gain  as the received signal power with increasing $N$. }

{Next, we show the performance  scaling law of an FD AF relay aided system under the same setup as the above IRS-aided system. The relay is equipped with $N$ transmit and $N$ receive antennas and the direct channel from the AP to the user can be similarly ignored when $N$ is asymptotically large. We assume that the relay adopts linear receive and transmit beamforming vectors, denoted by $\x^H_r$ and $\x_t$,  respectively.  In addition, perfect self-interference cancellation (SIC)  is assumed at the relay so that the obtained  performance serves as an upper bound for the practical case with imperfect SIC. In this case, the user receive   SNR can be expressed as
\begin{align}\label{SNR:AF}
\gamma_{FD} =\frac{P P_r\| \x_r^H\g  \|^2 \| \h^H_{r}\x_t  \|^2}{P_r \sigma^2_r\| \x_r \|^2\| \h^H_{r}\x_t \|^2 +  P\sigma^2\| \x_t \|^2\|\x^H_r\g \|^2+\sigma^2_r\sigma^2\| \x_t \|^2 \|\x_r \|^2},
\end{align}
where  $P_r$ and $\sigma^2_r$ denote the transmit power and the noise power at the relay, respectively.   It is not difficult to show that  the optimal solution maximizing $\gamma_{FD}$ satisfies $\x^*_t = \frac{\h_{r}}{\| \h_{r}\|}$ and $\x^*_r = \frac{\g}{\| \g\| }$. Substituting $\x^*_t$ and $\x^*_r$ into \eqref{SNR:AF}, we have
$\gamma_{FD} =    \frac{PP_r\|\g\|^2 \|\h_{r}\|^2  }{  P_r \sigma^2_r \|\h_{r}\|^2   + P \sigma^2 \|\g\|^2 +\sigma^2_r\sigma^2}.$
Then, we have the following proposition.
\begin{proposition}\label{scaling:law:AF}
Assume $\bm{h}^H_{r} \sim \mathcal{CN}(\bm{0},\varrho^2_h\I)$ and  $\bm{g} \sim \mathcal{CN}(\bm{0},\varrho^2_g\I)$.   As $N\rightarrow \infty$, it holds that
\begin{align}
\gamma_{FD} \rightarrow   \frac{PP_r \varrho^2_{g}\varrho^2_{h}N  }{  P_r  \sigma^2_r \varrho^2_{h} + P \sigma^2\varrho^2_{g} }.  \label{SNR2}
\end{align}
\end{proposition}
\begin{proof}
Since $\frac{\|\g\|^2}{N}\rightarrow \varrho^2_{g}$ and $\frac{\|\h_{r}\|^2}{N}\rightarrow \varrho^2_{h}$  as $N\rightarrow \infty$ \cite{Hien2013}, it follows that
\begin{align}\label{SNR:3}
\gamma_{FD}  \rightarrow    \frac{P P_r  \varrho^2_{g}\varrho^2_{h} N^2  }{  P_r \sigma^2_r \varrho^2_{h}N  +  P\sigma^2 \varrho^2_{g}N +\sigma^2_r\sigma^2}\approx \frac{PP_r \varrho^2_{g}\varrho^2_{h}N  }{  P_r  \sigma^2_r \varrho^2_{h} + P \sigma^2\varrho^2_{g} }.
\end{align}
This thus completes the proof.
\end{proof}   }

{Proposition  \ref{scaling:law:AF} shows that even with perfect SIC, the receive SNR by using the FD AF relay  increases only linearly with $N$ when $N$ is asymptotically large.  This is fundamentally due to the noise effect at the AF relay. To be specific, although the signal power in the FD AF relay system scales in the order of $N^2$ same as that in the IRS-aided system, its effective noise power at the receiver also scales linearly with $N$  (see  \eqref{SNR:3}) in contrast to the constant noise power $\sigma^2$ in the IRS-aided system, thus resulting in a lower SNR gain order with $N$. Last, it is worth mentioning that for the HD AF relay system, its receive SNR scaling order with $N$ can be shown to be identical to  that of the FD AF relay system given in Proposition 3.}
\section{Multiuser System}\label{multiuser:sec}
In this section, we consider the general multiuser setup.  Specifically, we propose two efficient algorithms to solve (P1) suboptimally  by generalizing the two approaches in the single-user case.
\subsection{Alternating Optimization Algorithm }
This algorithm leverages  the alternating optimization similarly as in the single-user case, while  the transmit beamforming at the AP is designed   by applying the well-known minimum mean squared error (MMSE) criterion to cope with the multiuser interference instead of using MRT in the single-user case without interference.
For given phase shift ${\bm \theta}$,  the combined channel from the AP to user $k$ is given by ${\bm{h}}^H_k= \bm{h}^H_{r,k}\ttheta\bm{G}+\bm{h}^H_{d,k}$.
Thus, problem (P1) is reduced to
\begin{align}
\text{(P3)}: ~~\min_{\W} ~~~&\sum_{k=1}^{K}\|\bm{w}_k\|^2 \\
\mathrm{s.t.}~~~~&\frac{|{\bm{h}}^H_k\bm{w}_k |^2}{\sum_{j\neq k}^{K}|{\bm{h}}^H_k\bm{w}_j |^2 +  \sigma^2_k}\geq \gamma_k, \forall k.\label{P2:SINR}
\end{align}
Note that (P3) is the conventional power minimization problem in the multiuser MISO downlink broadcast channel, which can be efficiently solved by using second-order cone program (SOCP) \cite{wiesel2006linear}, SDP \cite{bengtsson2001handbook},  or a fixed-point iteration algorithm based on the uplink-downlink duality \cite{schubert2004solution,luo2006introduction}. In addition, it is easy to verify that at the optimal solution to problem (P3), all the SINR constraints in \eqref{P2:SINR} are met with equalities.

On the other hand,  for given transmit beamforming $\W$,  problem (P1) is reduced to a feasibility-check problem. Let $\bm{h}^H_{d,k}\bm{w}_j ={b}_{k,j}$ and  $v_n=e^{j\theta_n}$, $n=1,\cdots, N$. By applying the change of variables  $\bm{h}^H_{r,k}\ttheta\bm{G}\bm{w}_j =\bm{v}^H\bm{a}_{k,j}$ where $\bm{v} = [e^{j\theta_1}, \cdots, e^{j\theta_N}]^H$ and $\bm{a}_{k,j}=\text{diag}(\bm{h}^H_{r,k})\bm{G}\bm{w}_j$,   problem (P1) is reduced to
\begin{align}\label{secIV:p3}
\text{Find}~~ &~~ \bm{v} \\
\mathrm{s.t.}~~&\frac{|\bm{v}^H\bm{a}_{k,k} + b_{k,k}  |^2}{\sum_{j\neq k}^{K}|\bm{v}^H\bm{a}_{k,j} + b_{k,j}  |^2 +  \sigma^2_k}\geq \gamma_k, \forall k, \label{eq:coupling2}\\
& |v_n|=1,  n=1,\cdots, N. \label{eq:modulus1}
\end{align}
{While the above problem appears similar to the relay beamforming optimization problem  for multi-antenna relay broadcast channel \cite{zhang2009joint}, it cannot be directly transformed into an SOCP optimization problem because  the phase rotation of the common vector $\vv$ may not render $\bm{v}^H\bm{a}_{k,k} + b_{k,k}$'s in \eqref{eq:coupling2} to be real numbers for all users.  Moreover,  it has non-convex unit-modulus constraints in \eqref{eq:modulus1}. However, by observing that constraints  \eqref{eq:coupling2} and \eqref{eq:modulus1} can be transformed into quadratic constraints, we apply the SDR technique to approximately solve  problem \eqref{secIV:p3} efficiently.}

Specifically, by introducing an auxiliary variable $t$, \eqref{secIV:p3} can be equivalently written as
\begin{align}\label{secIV:p4}
\text{Find} ~~~~&\bm{v}\\
\mathrm{s.t.}~~~~& \bm{\bar v}^H\R_{k,k}{\bm {\bar v}} + |b_{k,k}|^2 \geq \gamma_k\sum_{j\neq k}^{K}   \bm{\bar v}^H\R_{k,j}{\bm {\bar v}}  + \gamma_k(\sum_{j\neq k}^{K} |b_{k,j}|^2 + \sigma^2_k),    \forall k,\label{P10:qos}\\
& |v_n|^2=1,  n=1,\cdots, N+1,
\end{align}
where
\[
\bm{R}_{k,j}=\begin{bmatrix}
\bm{a}_{k,j}\bm{a}_{k,j}^H  & \bm{a}_{k,j}b^H_{k,j} \\
 \bm{a}_{k,j}^Hb_{k,j} & 0 \\
\end{bmatrix},~~
\bm{\bar{v}}=\begin{bmatrix}
\bm{v}  \\
t \\
\end{bmatrix}.
\]
Note that $\bm{\bar{v}}^H\bm{R}_{k,j}\bm{\bar{v}}={\rm{tr}}(\bm{R}_{k,j}\bm{\bar{v}}\bm{\bar{v}}^H)  $.
Define $\bm{V}=\bm{\bar{v}}\bm{\bar{v}}^H$, which needs to satisfy  $\bm{V}\succeq \bm{0}$ and ${\rm{rank}}(\bm{V})=1$. Since the rank-one constraint is non-convex, we relax this constraint and problem \eqref{secIV:p4} is then transformed to
\begin{align}
\text{(P4)}: ~~\text{Find} ~~~~&\bm{V}\\
\mathrm{s.t.}~~~~&{\rm{tr}}( {\bm{R}_{k,k}}{\bm V})  + |b_{k,k}|^2 \geq \gamma_k\sum_{j\neq k}^{K}   {\rm{tr}}( {\bm{R}_{k,j}}{\bm V})  + \gamma_k(\sum_{j\neq k}^{K} |b_{k,j}|^2 + \sigma^2_k),   \forall k,\label{P6:SINR:39}\\
~~~~& \bm{V}_{n,n} = 1,  n=1,\cdots, N+1, \label{P6:C9} \\
~~~~&\bm{V} \succeq 0.  \label{P6:C10}
\end{align}
It is not difficult to observe that problem (P4) is an  SDP and hence it can be optimally solved by existing convex optimization solvers such as CVX \cite{cvx}. While the SDR may not be tight for problem \eqref{secIV:p4},  the Gaussian randomization can be  similarly used  to obtain a feasible solution to problem \eqref{secIV:p4} based on the higher-rank solution obtained by solving (P4).
 In addition, it is worth pointing out that the SINR constraints in \eqref{P6:SINR:39} are not necessarily to be met with equality for a feasible solution of (P4), due to the common phase shifting matrix ($\V$) for all users.

 In the proposed  alternating optimization algorithm, we solve problem (P1) by solving problems (P3) and (P4) alternately in an iterative manner, where the solution obtained  in each iteration is used as the initial point  of the next iteration. The details of the proposed algorithm are summarized in Algorithm \ref{Alg:MMSE}.  In particular, the algorithm starts with solving problem (P3) for given $\bm{\theta}$ instead of solving (P4) for given  $\W$. This is deliberately designed since (P3) is always feasible for any arbitrary $\bm{\theta}$, provided that ${\rm{rank}}(\G^H\ttheta \HH_r + \HH_d)=K$, while this may not be true for (P4) with arbitrary $\W$.
On the other hand,  as solving (P4) only attains a feasible solution,  it remains unknown whether the objective value of (P3) will monotonically decrease or not over iterations in Algorithm \ref{Alg:MMSE}.  Intuitively, if the feasible solution obtained by solving (P4) achieves a strictly larger user SINR than the corresponding SINR target $\gamma_k$ for user $k$, then the transmit power of user $k$ and hence the total transmit power in problem (P3) can be properly reduced without violating all the SINR constraints. More rigorously, the convergence of Algorithm 1 is ensured  by the following proposition.
\begin{proposition}
The objective value of (P3) is non-increasing  over the iterations by applying Algorithm \ref{Alg:MMSE}.
\end{proposition}
\begin{proof}
Denote the objective value of (P3) based on a feasible solution $({\bm{\theta}}, \W)$ as $f({\bm{\theta}}, \W)$. As shown in step 4 of Algorithm 1, if there exists a feasible solution to problem (P4), i.e.,  $({\bm{\theta}}^{r+1}, \W^{r})$ exists, it is also feasible to problem (P3).  As such, $({\bm{\theta}}^r, \W^{r})$ and $ ({\bm{\theta}}^{r+1}, \W^{r+1})$ in step 3 are the feasible solutions to  (P3) in the $r$th and $(r+1)$th iterations, respectively. It then follows that
$f({\bm{\theta}}^{r+1}, \W^{r+1})\overset{(a)}{\geq} f({\bm{\theta}}^{r+1}, \W^r)\overset{(b)}{=}  f({\bm{\theta}}^{r}, \W^r)$,
where $(a)$ holds since for given ${\bm{\theta}}^{r+1}$ in step 3 of Algorithm 1, $\W^{r+1}$ is the optimal solution to problem (P3);  and $(b)$ holds because the objective function of (P3) is regardless  of ${\bm{\theta}}$ and only depends on $\W$.
\end{proof}
 \begin{algorithm}[t]
\caption{Alternating optimization algorithm.}\label{Alg:MMSE}
\begin{algorithmic}[1]
\STATE Initialize the phase shifts ${ \bm{\theta}}={\bm{\theta}}^1$ and set the iteration number $r=1$.
\REPEAT
\STATE Solve problem (P3) for  given ${\bm{\theta}}^r$, and denote the optimal solution as $\W^{r}$.
\STATE {Solve problem (P4) or (P4')  for given $\W^{r}$, and denote the solution after performing Gaussian randomization as ${ {\bm{\theta}}^{r+1}}$.}
\STATE Update $r=r+1$.
\UNTIL{  The fractional decrease of the objective value  is below a threshold $\epsilon>0$ or problem (P4)/(P4') becomes  infeasible.}
\end{algorithmic}
\end{algorithm}

To achieve better converged solution, we further transform problem (P4) into an optimization problem with an explicit objective to obtain a generally more efficient  phase shift solution to reduce the transmit power. The rationale is that for the transmit beamforming optimization problem, i.e., (P3), all the SINR constraints are active at the optimal solution. As such, optimizing the phase shift to enforce the user achievable  SINR to be larger than the SINR target in  (P4) directly leads to the transmit power reduction in  (P3) (e.g., by simply scaling down the power of transmit beamforming). To this end, problem (P4) is transformed into the following problem
\begin{align}
 \text{(P4')}:  ~\max_{{\bm V}, \{\alpha_k\},  } ~&\sum_{k=1}^{K} \alpha_k \\
\mathrm{s.t.}~~&{\rm{tr}}( {\bm{R}_{k,k}}{\bm V})  + |b_{k,k}|^2 \geq \gamma_k\sum_{j\neq k}^{K}   {\rm{tr}}( {\bm{R}_{k,j}}{\bm V})  + \gamma_k(\sum_{j\neq k}^{K} |b_{k,j}|^2 + \sigma^2_k) + \alpha_k,   \forall k,\label{P6:SINR}\\
~~& \bm{V}_{n,n} = 1,  n=1,\cdots, N+1, \label{P6:C9} \\
~~&\bm{V} \succeq 0,  \alpha_k \geq  0,   \forall k, \label{P6:C10}
\end{align}
where the slack variable $\alpha_k$ can be interpreted as the ``SINR residual'' of user $k$ in phase shift optimization.
Note that (P4) and (P4') have the same set of feasible   $\V$, while (P4') is more efficient than (P4) in terms of the converged solution, as will be verified in Section V-B by simulation. 

%
\vspace{-0.4cm}
\subsection{Two-Stage Algorithm }
Inspired by the combined channel gain maximization problem \eqref{secIII:p3} in the single-user case, we next propose a two-stage algorithm with lower complexity compared to the alternating optimization algorithm  by decoupling the joint beamforming design problem (P1) into two beamforming subproblems, for optimizing the phase shifts and transmit beamforming, respectively.  Specifically,  the phase shifts at the IRS are optimized in the first stage by solving a weighted effective channel gain maximization problem. This aims to align with the phases of different user channels so as to maximize the beamforming gain of the IRS, especially for the users near to the IRS.  In the second stage, we solve problem (P3) to obtain the optimal MMSE-based transmit beamforming with given phase shifts $ \bm{\theta}$.

Let  $\bm{v} = [e^{j\theta_1}, \cdots, e^{j\theta_N}]^H \in \mathbb{C}^{N \times 1}$ and $\bm{\Phi}_k =  \text{diag}(\bm{h}^H_{r,k})\bm{G}  \in \mathbb{C}^{N \times M}$, $\forall k$. The weighted sum of the combined channel gain of all users is expressed as
\begin{align}\label{eq:obj}
\sum_{k=1}^{K}t_k\| \bm{h}^H_{r,k}\ttheta\bm{G}+\bm{h}^H_{d,k}  \|^2= &\sum_{k=1}^{K}t_k\|\vv^H \text{diag}(\bm{h}^H_{r,k})\bm{G}  + \bm{h}^H_{d,k}  \|^2 = \sum_{k=1}^{K}t_k\|\vv^H \bm{\Phi}_k + \bm{h}^H_{d,k}  \|^2,  
  \end{align}
  where we set the weights to be $t_k = \frac{1}{ \gamma_k\sigma^2_k}, k=1,...,K$, motivated by constraint \eqref{SecIII:SNRconstraint0}. Based on \eqref{eq:obj}, the phase shifts can be obtained by solving the following problem
\begin{align}
\text{(P5)}: ~~\max_{\vv} ~~~&\sum_{k=1}^{K}t_k\|\vv^H \bm{\Phi}_k + \bm{h}^H_{d,k}  \|^2 \\
\mathrm{s.t.}~~~~&  |v_n|=1,  n=1,\cdots, N. \label{eq:modulus2}
\end{align}
Note that for $K=1$,  (P5) is equivalent to problem \eqref{secIII:p3} for the single-user case in Section III-A.
However, in the multiuser case, due to the same set of phase shifts applied  for all users with different channels, the combined channel power gains of different users cannot be maximized at the same time in general, which thus need to be balanced for optimally solving (P5). Nevertheless, since problem (P5) is a non-convex QCQP, it can be similarly reformulated as a homogeneous QCQP as in Section III-A and then solved by applying the SDR and Gaussian randomization techniques. The details are omitted here for brevity. With the  phase shifts obtained from (P5), the MMSE-based transmit bemaforming is then obtained by solving (P3).
{Compared to the alternating optimization based  algorithm proposed in Section \ref{multiuser:sec}-A, the two-stage algorithm has lower computational complexity as (P5) and (P3) only need to be respectively solved for one time, but may suffer from certain performance loss, which will be evaluated in the next section.}


\section{Simulation Results}\label{simulation:sec}
We consider a three-dimensional (3D) coordinate system where a uniform linear array (ULA) at the AP and a uniform rectangular array (URA) at the IRS are located in $x$-axis and $x$-$z$ plane, respectively. The antenna spacing is half wavelength and the reference (center) antennas at the AP and IRS are respectively  located at $(0, 0, 0)$ and  $(0, d_0, 0)$, where $d_0>0$ is the distance between them.  For the IRS, we set $N=N_{x}N_{z}$ where $N_{x}$ and $N_{z}$ denote the numbers of reflecting elements along the $x$-axis and  $z$-axis, respectively.   For the purpose of exposition, we fix $N_x=5$ and increase $N_z$ linearly with $N$.  The distance-dependent path loss model is given by
\begin{align}\label{pathloss}
L(d) = C_0\left( \frac{d}{D_0} \right)^{-\alpha},
\end{align}
where $C_0$ is the path loss at the reference distance $D_0=1$ meter (m), $d$ denotes the individual link distance, and $\alpha$ denotes the path loss exponent.
To account for  small-scale fading, we assume the Rician fading channel model for all channels involved.  Thus, the AP-IRS channel $\G$ is given by
 \begin{align}
 \G = \sqrt{\frac{ \beta_{\rm AI} }{1+  \beta_{\rm AI} }}\G^{\rm LoS} +  \sqrt{\frac{1}{1+\beta_{\rm AI}}}\G^{\rm NLoS},
 \end{align}
 where $\beta_{\rm AI}$ is the Rician factor, and $\G^{\rm LoS} $ and $\G^{\rm NLoS}$ represent  the deterministic LoS (specular) and Rayleigh fading components, respectively. In particular, the above model is reduced to the  LoS channel when $ \beta_{\rm AI} \rightarrow  \infty$ or  Rayleigh fading channel when  $ \beta_{\rm AI} = 0$.
  The elements in $ \G$ are then multiplied by the square root of the distance-dependent path loss in \eqref{pathloss} with the path loss exponent denoted by   $\alpha_{\rm AI}$. The AP-user and IRS-user channels are also generated by following the similar procedure.  The path loss exponents of the AP-user and IRS-user links are  denoted by $\alpha_{\rm Au}$ and  $\alpha_{\rm Iu}$, respectively, and the Rician factors of the two links are denoted by  $\beta_{\rm Au}$ and $\beta_{\rm Iu}$, respectively.
  Due to the relatively large  distance and random scattering  of the AP-user channel, we set $\alpha_{\rm Au}=3.5$ and $\beta_{\rm Au}=0$ while their counterparts for AP-IRS and IRS-user channels will be specified later to study their effects on the system performance.
  Without loss of generality, we assume that all users have the same SINR target, i.e., $\gamma_k=\gamma$, $\forall k$. {The number of random vectors used for the Gaussian randomization is set to be 1000 and the stopping threshold for the alternating optimization algorithms  is set as $\epsilon=10^{-4}$.}  Other system parameters are set as follows: $C_0= -30$ dB,  $\sigma_k^2=-80$\,dBm, $\forall k$,  and  $d_0= 51$ m (if not specified otherwise).

\vspace{-0.4cm}
\subsection{Single-User System}

\begin{figure*}[!t]
\centering
\includegraphics[width=3in, height=1in]{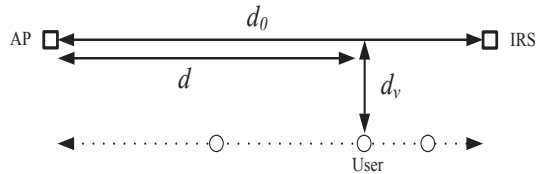}  
\caption{Simulation setup of the single-user case (top view).}
 \label{simulation:setup}
\end{figure*}


First, we consider a single-user system with the user SNR target $\gamma=10$ dB and $M=4$. As shown in Fig. \ref{simulation:setup}, the user lies on a horizontal line that is in parallel to the one that connects the reference antennas of AP and IRS, with the vertical distance between these two lines being  $d_v = 2$ m. Denote the horizontal distance between the AP and  user  by $d$ m. Accordingly, the AP-user and IRS-user link distances are given by $d_1= \sqrt{d^2+d_v^2}$ and $d_2= \sqrt{(d_0-d)^2+d_v^2}$, respectively.  By varying the value of $d$, we can study the transmit power required for serving the user located between the AP and  IRS, under the given SNR target.
The path loss exponents and Rician factors are set as $\alpha_{\rm AI} =2$, $\alpha_{\rm Iu}=2.8$, $\beta_{\rm Iu}=0$, and  $\beta_{\rm AI} = \infty$, respectively, where  $\G$ is of rank one, i.e., an LoS channel between the AP and IRS.
We compare the following  schemes: 1) Lower bound: the minimum transmit power based on the optimal solution of the SDP problem \eqref{secIII:p6}; 2) SDR: the solution obtained by applying SDR and Gaussian randomization techniques in Section III-A; 3) Alternating optimization: the solution proposed in Section III-B; 4) AP-user MRT: we set $\bm{\bar w} = {{\bm{h}_d}}/{\|{\bm{h}_d}\|}$ to achieve MRT based on the AP-user direct channel;
 5) AP-IRS MRT: we set $\bm{\bar w} ={{\bm{g}}}/{\|\bm{g}\|}$ to achieve MRT based on the AP-IRS  rank-one channel, with $\bm{g}^H$ denoting any row in $\G$;  6) Random phase shift: we set  the elements in $\bm{\theta}$ randomly in $[0, 2\pi]$ and then perform MRT at the AP based on the combined channel;
  7) Benchmark scheme without the IRS by setting  $\bm{w} = \sqrt{\gamma \sigma^2}{\bm{h}_d}/{\|\bm{h}_d\|^2}$. Note that for scheme 3), the transmit beamforming is initialized by using the  AP-user MRT, and for schemes 4) and 5) with given $\bm{\bar w}$, the transmit power and phase shifts are optimized by using the results in Section III-B.

\subsubsection{AP Transmit Power versus AP-User Distance}
\begin{figure*}[!t]
\centering
\includegraphics[width=0.6\textwidth]{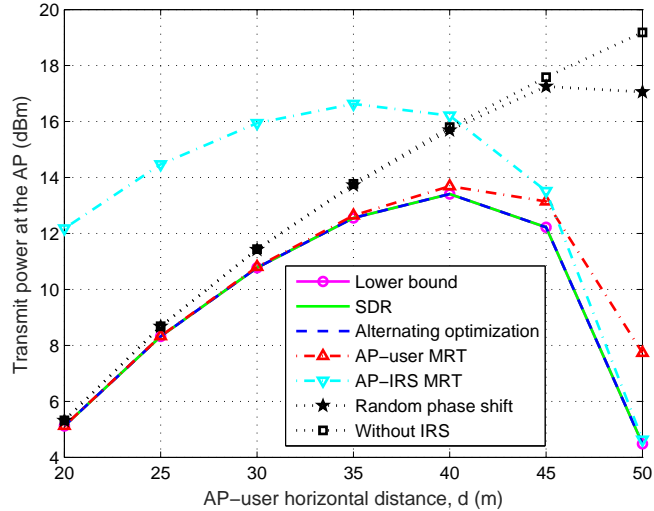}
\caption{AP transmit power versus AP-user horizontal distance}
\label{simulation:distance}
\end{figure*}

%
 In Fig. \ref{simulation:distance}, we compare the  transmit power required by all schemes versus the horizontal distance between the AP and user, $d$. First, it is observed that the proposed two schemes both achieve near-optimal transmit power as compared to the transmit power  lower bound, and also significantly outperform other benchmark schemes. Second, for the scheme without the IRS, one can observe that  the user farther away from the AP requires higher transmit power at the AP due to the larger signal attenuation.
  However, this problem is alleviated by deploying an IRS, which implies that a larger AP-user distance does not necessarily lead to a higher transmit power in IRS-aided wireless networks. This is because the user farther away from the AP may be closer to the IRS and thus it is able to receive stronger reflected signal from the IRS. As a result, the user near either the AP (e.g., $d=23$ m)  or IRS (e.g., $d=47$ m) requires  lower transmit power than a user far away from both of them (e.g., $d=40$ m).
  This phenomenon also suggests that the signal coverage can be effectively extended by deploying only a passive IRS rather than installing an additional AP or active relay. For example, for the same transmit power about 13 dBm, the coverage of the network without the IRS is about 33 m whereas this value is improved to be beyond 50 m by applying the proposed joint beamforming designs with an IRS.

  On the other hand, it is observed from Fig. \ref{simulation:distance} that the AP-user MRT scheme performs close to optimal when the user is nearer to the AP, while it incurs considerably higher transmit power when the user is nearer to the IRS. This is expected since in the former case, the user received signal is dominated by the AP-user direct link whereas the IRS-user link is dominant in the latter case. Moreover, it can be observed that the AP-IRS MRT scheme behaves oppositely as the user moves away from the AP toward IRS.  Finally, Fig. \ref{simulation:distance} also shows  that if the transmit beamforming is not designed properly, the performance achieved by using the IRS may be even worse than that of the case without the IRS, e.g., with the AP-IRS MRT scheme for $d\leq 35$ m. This further demonstrates that the proposed joint beamforming designs  can dynamically adjust the AP's beamforming to strike an optimal  balance between the signal power transmitted  directly  to  the user and  that to the IRS, to achieve the maximum received power at the user.

\subsubsection{AP Transmit Power versus Number of Reflecting Elements}
\begin{figure}[ht]\center
\subfigure[$d = 50$ m]{\includegraphics[width=0.495\textwidth]{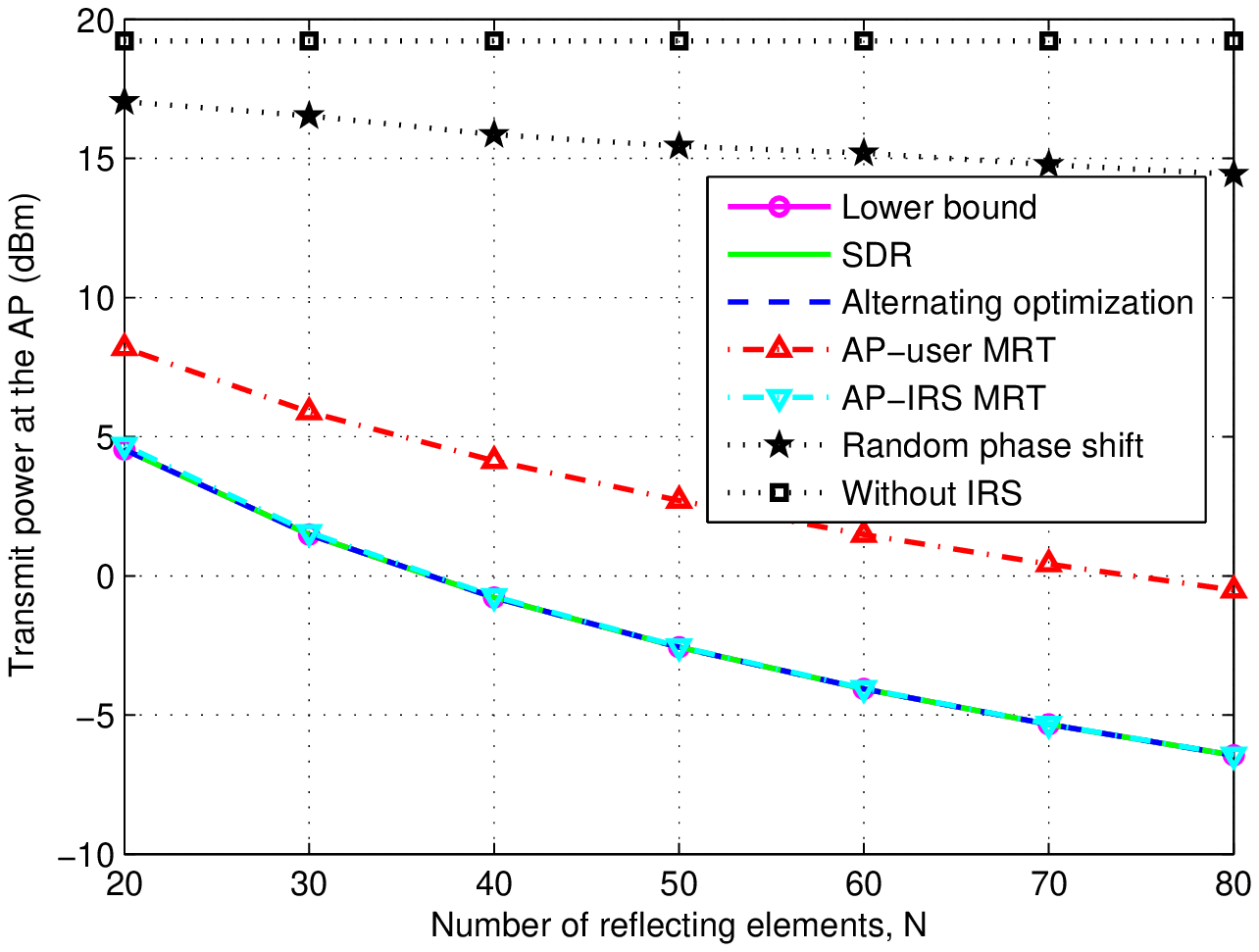}}
\subfigure[$d = 41$ m]{\includegraphics[width=0.495\textwidth]{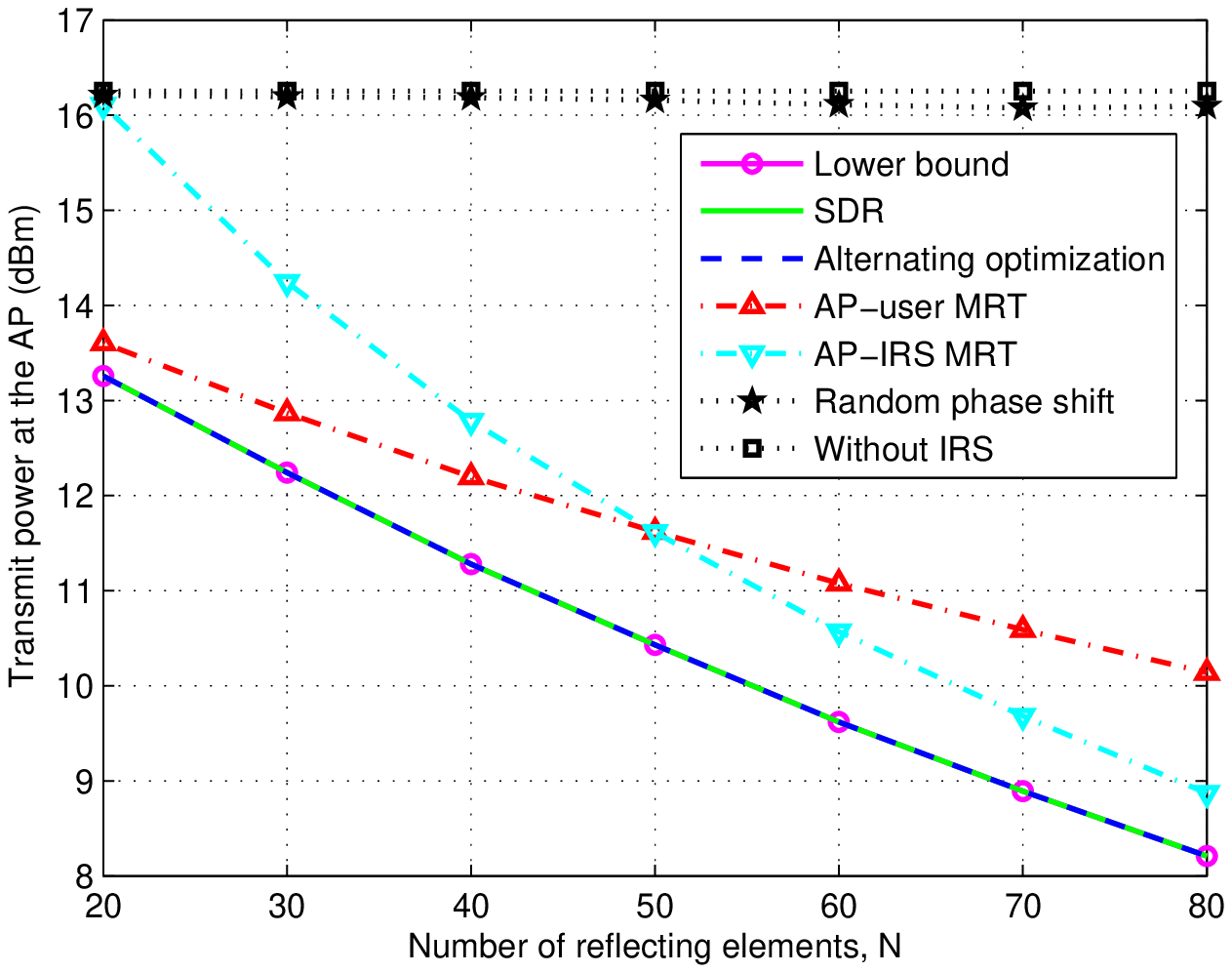}} 
\subfigure[$d = 15$ m]{\includegraphics[width=0.495\textwidth]{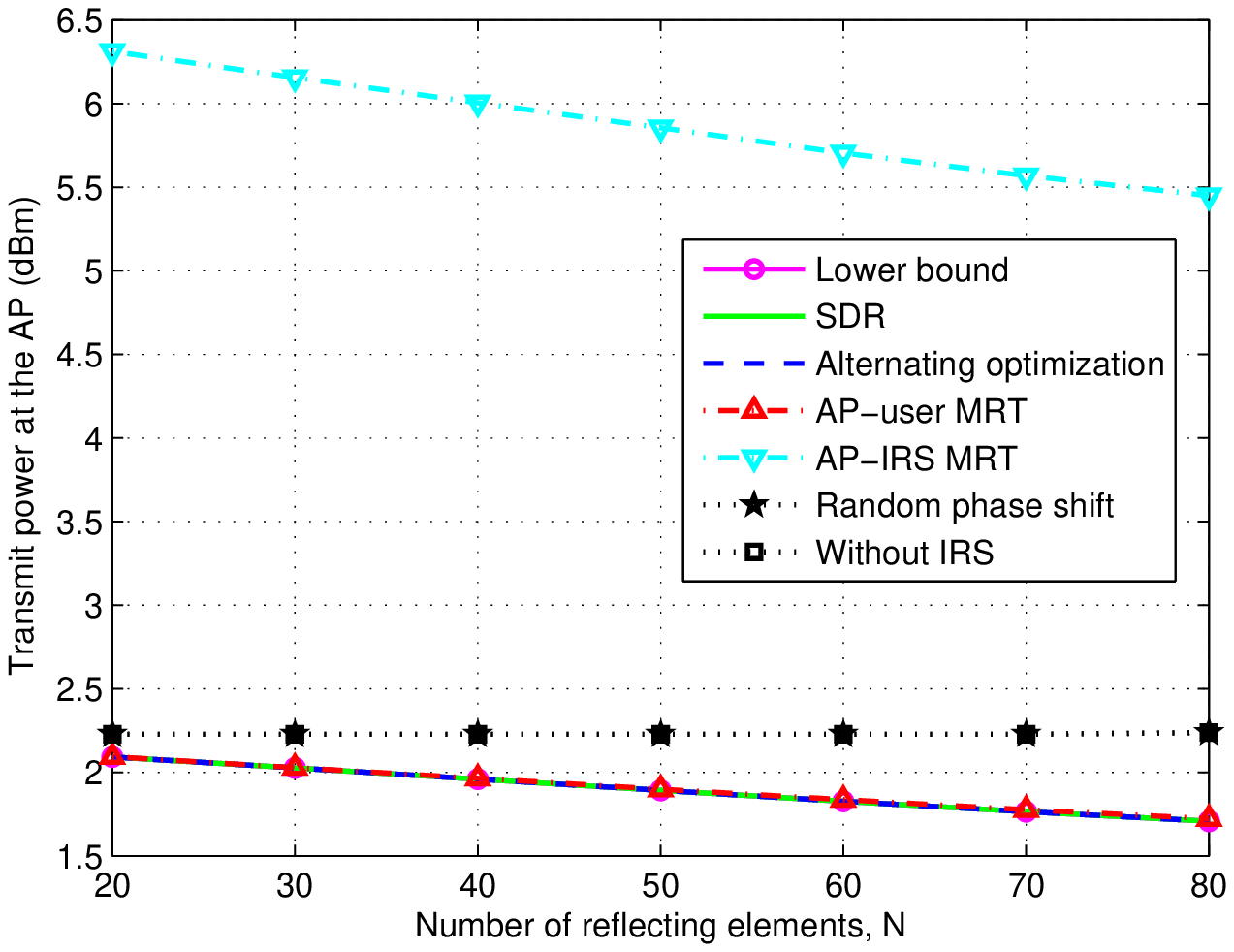}} 
\caption{AP transmit power versus the number of reflecting elements at the IRS, $N$. } \label{simulation:N} \vspace{-6mm}
\end{figure}

In Fig. \ref{simulation:N}, we compare the AP's transmit power of all the above schemes versus the number of reflecting elements at the IRS when $d=50$, $41$, and $15$ m, respectively. From Fig. \ref{simulation:N} (a), it is observed that for the case of $d=50$ m (i.e., the user is very close to the IRS), the AP-IRS MRT scheme achieves near-optimal transmit power since the signal reflected by the IRS is much stronger than that directly from the AP at the user. Furthermore, it is interesting to note that  the transmit power required by the proposed schemes scales down with the number of reflecting elements $N$ approximately in the order of  $N^2$ in this case, which is in accordance with Proposition \ref{scaling:law} (even when the AP-IRS channel is LoS rather than Rayleigh fading). For example, for the same user SNR,  a transmit power of 2 dBm is required at the AP when $N=30$ while this value is reduced to $-4$ dBm when  $N=60$, which suggests an around  6 dB gain by doubling the number of reflecting elements.
  In contrast, the transmit power required  by using the random phase shift  decreases with increasing $N$  in a much slower rate,  because without reflect beamforming the average signal power of the reflected signal is comparable to that of the signal from the AP-user direct link in this case.  Finally, it is observed that the above  gains diminish as the user moves away from the IRS. For example, for the case of $d=15$ m shown in Fig. \ref{simulation:N} (c)  where the AP-user direct link signal is much stronger than that of the IRS-user link, the required transmit power is insensitive to the number of reflecting elements.
For the case of $d=41$ m shown in Fig. \ref{simulation:N} (b) when the user is neither close to the AP nor close to the IRS, it is observed that the transmit power gain of the proposed schemes  is generally lower than $N^2$. This is because in this case the signal power received at the IRS is compromised as the AP transmit beamforming is steered to strike a balance between the AP-IRS link and the AP-user direct link. In practice, the number of reflecting elements can be properly selected depending on the IRS's location as well as the target user SNR/AP coverage range.
\begin{figure*}[!t]
\centering
\includegraphics[width=0.6\textwidth]{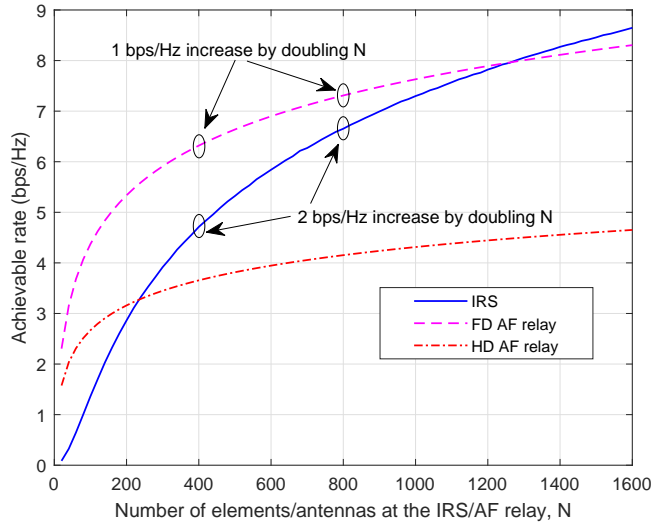}  
\caption{{Comparison between the IRS and the FD/HD AF relay.} }
 \label{simulation:AFrelay}
\end{figure*}

\subsubsection{Comparison with AF Relay}
{Next, we compare the achievable rates of the IRS versus the FD AF relay based on the results derived in Section III. We consider the setup in Fig. \ref{simulation:setup} with $d_0=d=100$ m, $d_v=1$ m, ${\sigma}^2_r=-80$ dBm, $\alpha_{\rm AI} =3.2$, $\alpha_{\rm Iu}=2$, $\beta_{\rm AI} = 0$,  $\beta_{\rm Iu}=\infty$, and $M=1$. To focus on the comparison with large $N$, the direct link from the AP to the user is ignored, and perfect SIC is assumed for the FD AF relay. As such, the SNRs and the corresponding achievable rates for the IRS-aided  and the FD AF relay-aided systems can be obtained based on \eqref{IRS:receive:power} and \eqref{SNR:AF}, respectively. For a fair comparison, we assume that both systems have the same total transmit power budget $P=5$ mW (for single link only). Since the IRS is passive, all the transmit power is used at the AP, whereas since the AF relay is active like the AP, an optimal power allocation between them is required which can be obtained by exhaustive search. }

{Under the above setup, we plot the achievable rate in bits/second/Hertz (bps/Hz) versus $N$ in Fig. \ref{simulation:AFrelay}, where the HD AF relay is also considered as a benchmark. It is observed that when $N$ is small, the IRS-aided system is able to achieve the same rate as the FD/HD AF relay-aided system by using more reflecting elements. However, since the IRS's elements are passive, no transmit RF chains are needed for them and thus the cost is much lower as compared with that of active antennas for the AF relay requiring transmit RF chains. Furthermore, one can observe that by doubling $N$ from 400 to 800, the achievable rate of using IRS increases about 2 bps/Hz whereas that of using the FD AF relay only increases about 1 bps/Hz. This is due to their different SNR gains ($N^2$ versus $N$) with increasing $N$ as revealed in Propositions 2 and 3. As a result, it is expected that the IRS-aided system will eventually outperform the FD/HD AF relay-aided system when $N$ is sufficiently large, as shown in Fig. \ref{simulation:AFrelay}.}
\subsection{Multiuser System}
\begin{figure*}[!t]
\centering
\includegraphics[width=3.3in, height=1.6in]{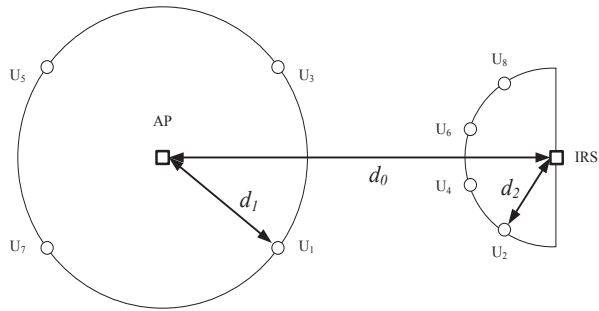}
\caption{Simulation setup of the multiuser case  (top view).}\label{simulation:MU}
\end{figure*}
 Next, we consider a multiuser system with eight users, denoted by $U_k$, $k=1,\cdots, 8$, and their locations are shown in Fig. \ref{simulation:MU}. Specifically, $U_k$'s, $k=2,4, 6, 8$, lie evenly on a half circle centered at  the reference antenna  of the IRS  with radius $d_2=3$ m, which are usually  considered as ``cell-edge'' users, as compared to $U_k$'s, $k=1,3, 5, 7$, which lie evenly on a circle centered at  the  reference antenna  of the AP  with radius $d_1=20$ m. Since the IRS can be practically  deployed in LoS with the AP and ``cell-edge'' users,  we  set   $\alpha_{\rm AI} =\alpha_{\rm Iu}=2.8$, $\beta_{\rm AI} =\beta_{\rm Iu}=3$ dB, respectively. We compare our proposed two algorithms (named as Alternating optimization w/ IRS and Two-stage algorithm w/ IRS, respectively) in Section IV with the two conventional designs in the case without the IRS, i.e.,  MMSE and zero-forcing (ZF) based beamforming \cite{bengtsson2001handbook,yoo2006optimality}. Specifically, the transmit power of the MMSE-based scheme without the IRS is obtained by solving (P3) with $\ttheta = {\bm 0}$, while that of the ZF-based scheme without the IRS is given by ${\rm{tr}}(\PP ( \HH_d^H \HH_d)^{-1})$ where $\PP = \text{diag}( \sigma^2_1\gamma_1,\cdots, \sigma^2_K\gamma_K)$. The transmit power required  by using the random phase shift at the IRS and MMSE beamforming at the AP is also plotted as a benchmark.
  Before comparing their performances, we first show the convergence behaviour of the proposed Algorithm \ref{Alg:MMSE} in Fig. \ref{convergence} by setting $M=4$ and considering that only $U_k$, $k=1,2,3,4$, are active  (need to be served) with  $\gamma=20$ dB.  The phase shifts are initialized using the  two-stage algorithm. It is observed  that the transmit power required  by the proposed algorithm decreases quickly with the number of iterations and solving (P4') instead of (P4) in Algorithm \ref{Alg:MMSE} achieves lower converged power. Thus, in the following simulations, we use (P4') instead of (P4) in Algorithm 1.


\begin{figure*}[!t]
\begin{minipage}[t]{0.5\linewidth}
\centering
\includegraphics[width=3.3in, height=2.4in]{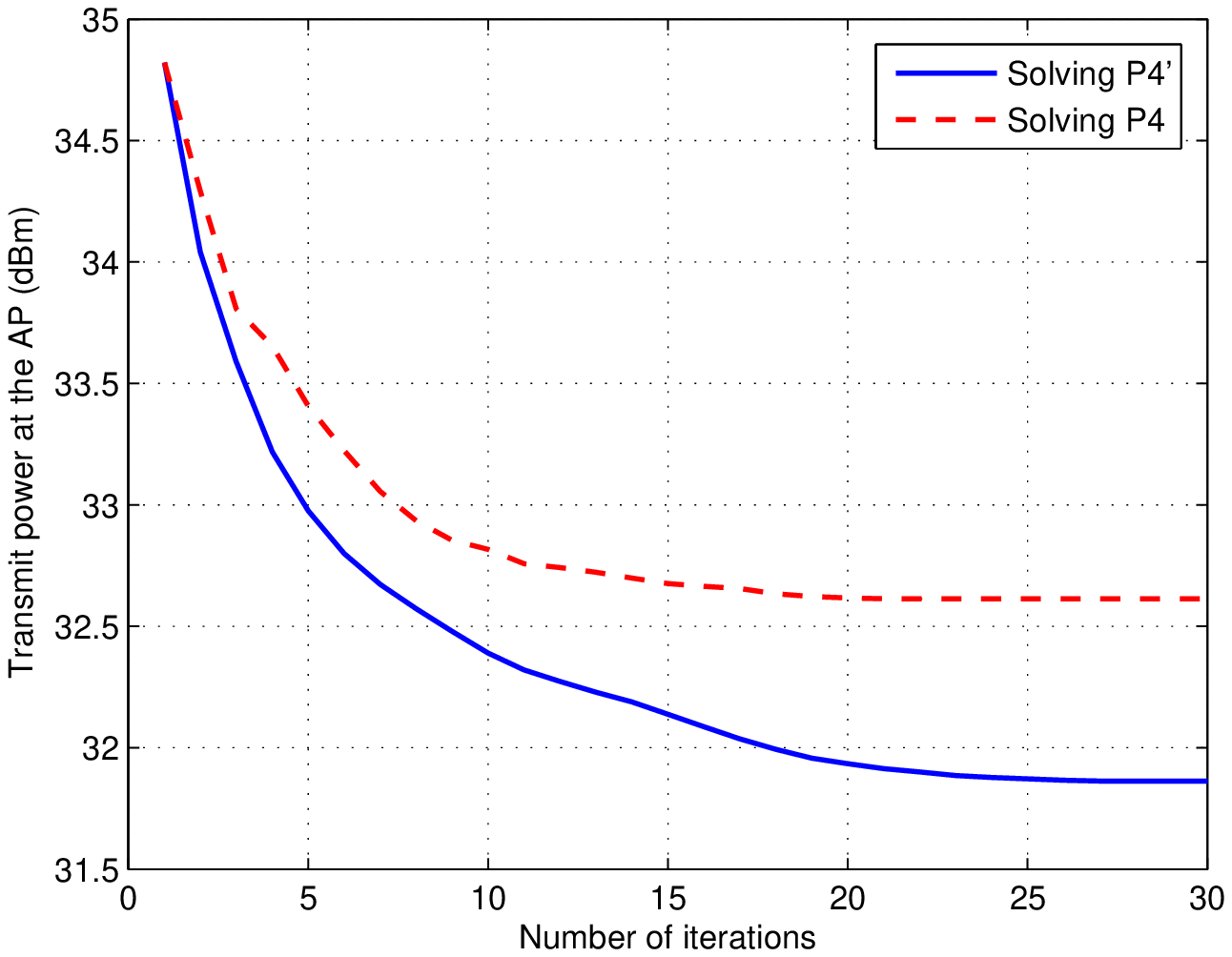}  
\caption{Convergence behaviour of the proposed alternating optimization  algorithm.}
 \label{convergence}
\end{minipage}%
~~~~\begin{minipage}[t]{0.5\linewidth}
\centering
\includegraphics[width=3.3in, height=2.4in]{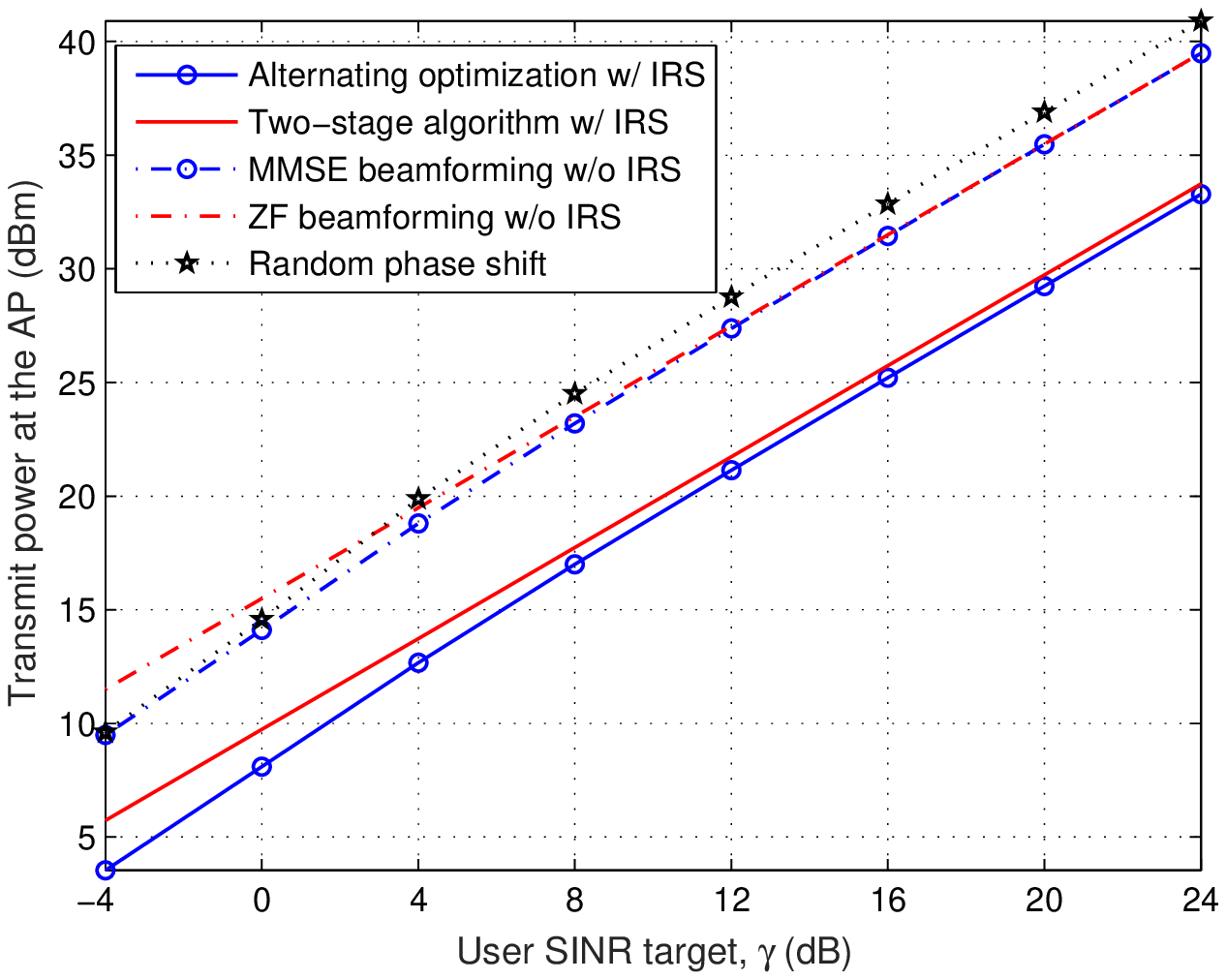} 
\caption{AP's transmit power versus the SINR target for the two-user case.}
\label{transmit:pow}
\end{minipage}\vspace{-8mm}
\end{figure*}


\begin{figure}[t]
\centering
\subfigure[Signal and interference powers at user 1 (far from the IRS).]{\includegraphics[width=0.415\textwidth]{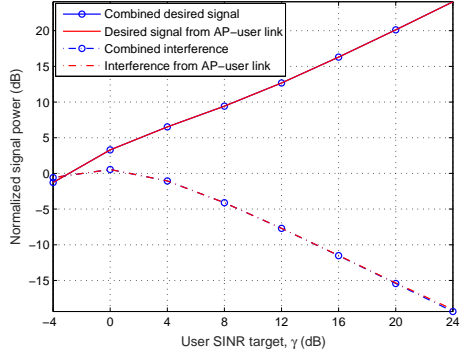}}
\subfigure[Signal and interference powers at user 2 (near the IRS). ]{\includegraphics[width=0.415\textwidth]{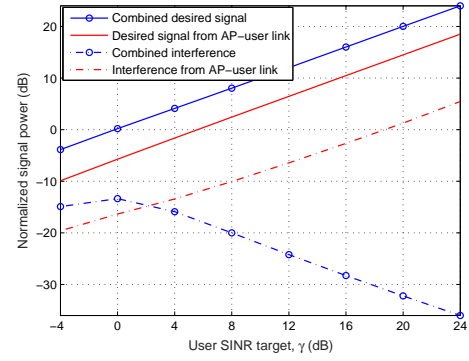}} 
\subfigure[$\rho_1$ ]{\includegraphics[width=0.415\textwidth]{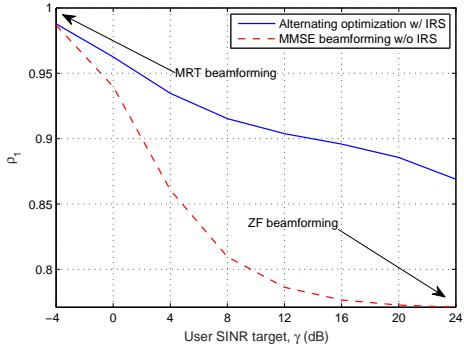}} 
\subfigure[ $\rho_2$ ]{\includegraphics[width=0.415\textwidth]{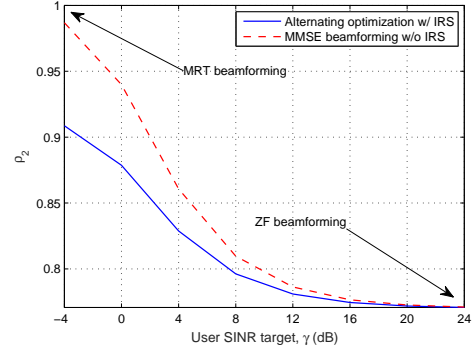}} 
\caption{Illustration for the collaborative working mechanism of the AP and IRS. } \label{simulation:pow2} \vspace{-8mm}
\end{figure}

\subsubsection{AP Transmit Power versus User SINR Target}
In Fig. \ref{transmit:pow}, we show the transmit power at the AP versus the SINR target by considering that only two users, namely $U_k$'s, $k=1,2$, are active, which are far from and near the IRS, respectively. It is observed that by adding the IRS,  the transmit power required by applying the proposed algorithms is significantly reduced, as compared to the case without the IRS. In addition, one can observe that the transmit power of the proposed two-stage algorithm asymptotically approaches that of the alternating optimization  algorithm as $\gamma$ increases.

Next, we show how the AP and IRS serve the two users by collaboratively steering the transmit beamforming and adjusting the phase shifts based on the proposed Algorithm \ref{Alg:MMSE} (i.e., Alternating optimization w/ IRS). To visualize the transmit beamforming direction at the AP, we use the effective angle between the transmit beamforming (i.e., $\bm{w}_k$) and the AP-user direct channel (i.e., $\bm{h}^H_{d,k}$)  for user $k$, which is defined as \cite{yoo2006optimality,tse2005fundamentals}
 $\rho_k\triangleq \mathbb{E} \left( \frac{| \bm{h}^H_{d,k}\bm{w}_k|}{\|\bm{h}^H_{d,k}\|\|\bm{w}_k\|}\right), \forall k.$
In particular, $\rho_k=1$ when the AP steers $\bm{w}_k$ to align with $\bm{h}^H_{d,k}$ perfectly, i.e., MRT transmission, whereas in general  $\rho_k<1$ when the AP sets $\bm{w}_k$ orthogonal to $\bm{h}^H_{d,j}$, $j\neq k$, i.e., ZF transmission.
Generally, a higher $\rho_k$ implies that the transmit beamforming direction is closer to that of the AP-user direct channel.  Figs. \ref{simulation:pow2} (a) and (b) show both the desired signal power and interference power of the two users, respectively. For each user $k$, we plot the following four power terms: 1) Combined desired signal power, i.e., $|( \bm{h}^H_{r,k}\ttheta \bm{G}+\bm{h}^H_{d,k})\bm{w}_k |^2$; 2) Desired signal power from AP-user link, i.e., $|\bm{h}^H_{d,k}\bm{w}_k |^2$; 3) Combined interference power, i.e.,  $|\bm{h}^H_{r,k}\ttheta \bm{G}+\bm{h}^H_{d,k})\bm{w}_j|^2$, $j\neq k$; 4) Interference power from AP-user link, i.e.,  $ |\bm{h}^H_{d,k}\bm{w}_j|^2$, $j\neq k$.
 For the purpose of exposition, both the desired signal and interference powers are normalized by the noise power. It is observed from Fig. \ref{simulation:pow2} (a)   that for user 1 (far from IRS), the desired signal and interference powers received from the combined channel are almost the same as those from the AP-user direct link. This is expected since user 1 is far away from the IRS and hence the signal (both desired and interference)  from the IRS is negligible. However, the case of user 2 (near IRS)  is quite different from that of user 1. As shown in Fig. \ref{simulation:pow2} (b), the desired signal power from the combined channel is remarkably higher than that from the AP-user direct link, due to the reflect beamforming gain by the IRS. Furthermore, the IRS also helps align the interference from the AP-IRS-user reflect link oppositely with that from the AP-user direct link to suppress the interference at user 2. This is shown  by observing that the interference power from the AP-user direct link monotonically  increases with the SINR target whereas that from the combined channel decreases  as SINR target increases.


In Figs. \ref{simulation:pow2} (c) and (d), we plot $\rho_1$ and $ \rho_2$, respectively, for both the cases with and without the IRS.
For user 1, it is observed in Fig. \ref{simulation:pow2} (c)  that for low SINR (noise-limited) regime,  the AP in the case with the IRS steers its beam direction toward the AP-user direct channel as in the case without the IRS (i.e., MRT beamforming), while for high SINR  (interference-limited) regime where the beam direction in the latter case is adjusted to null out the interference to user 2 (i.e, ZF beamforming), the former case still keeps a high correlation between user 1's beam direction and the corresponding AP-user direct channel. This is expected since the IRS is not able to enhance  the desired signal for user 1 (see Fig.  \ref{simulation:pow2} (a)),  thus keeping a high $\rho_1$ helps reduce the transmit power for serving user 1. However,   user 2 inevitably suffers more  interference from user 1 in  the AP-user direct link (see  Fig.  \ref{simulation:pow2} (b)), which, nevertheless, is significantly suppressed  at user 2, thanks to the IRS-assisted interference cancellation.  In contrast, from Fig. \ref{simulation:pow2} (d) it is observed that   the user 2's beam direction in the case with the IRS is steered toward the AP-user combined channel rather than the AP-user direct channel at low SINR regime, and then converges  to the  same ZF beamforming as in the case without the IRS at high SINR regime. This is expected since the IRS is not able to help cancel the AP-user direct interference at user 1, thus nulling out the interference caused by user 2 is the most effective way for meeting the high SINR target.
From the above, it is concluded that the transmit beamforming directions  for the users with the IRS are drastically different from those in the case without the IRS, depending on their different distances with the IRS.
The above results  further demonstrate the necessity of jointly optimizing the transmit beamforming and phase shifts in IRS-aided  multiuser systems.
  \begin{figure}[ht]
\centering
\includegraphics[width=0.55\textwidth]{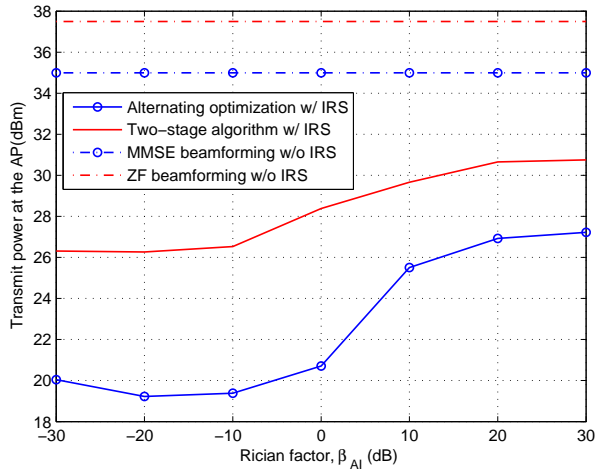}\vspace{-5mm}
\caption{AP's transmit power versus the Rician factor of the AP-IRS channel $\G$.} \label{Figures:K} \vspace{-3mm}
\end{figure}

\subsubsection{AP Transmit Power versus Rician Factor of $\G$}
In Fig. \ref{Figures:K}, we plot the AP's transmit power required by the proposed algorithms with IRS against the benchmark schemes without IRS versus the Rician factor of the AP-IRS channel by assuming all eight users are active and setting $M=8$, $N=40$, and $\gamma = 10$ dB. One can observe that when the average power of the LoS component is comparable to that of the fading component (e.g., $\beta_{\rm AI} \geq -10$ dB),  the transmit power of both  proposed algorithms increases with the increasing Rician factor,  for meeting the same set of SINR targets. This is because for users near the IRS, the reflected signals generally dominate the signals from AP-user direct links. As such, a higher Rician factor of $\G$ results in higher  correlation among the combined channels of these   users, which  is detrimental to the spatial multiplexing gain due to the more severe multiuser interference. The implication of this result is that (somehow surprisingly) it is practically favorable  to deploy the IRS in a relatively rich scattering environment to avoid  strong LoS (low-rank $\G$) with  the AP,  so as to serve multiple users by the IRS which requires sufficient spatial degrees of freedom from the AP to the IRS.

\subsubsection{{Comparison with Massive MIMO}}
{To compare with the existing TDD-based  massive MIMO system (without the use of IRS),  we consider  there are sixteen users with eight users randomly distributed within  $60$ m from the AP and the other eight users randomly distributed within $6$ m from the IRS. Other channel parameters are set to be the same as those for Fig. \ref{transmit:pow}. To facilitate the channel estimation, we assume that the IRS is equipped with receive RF chains. The transmission protocol for the IRS-aided wireless system is described  as follows: 1) all the users send orthogonal pilot signals concurrently as in the TDD-based massive MIMO system; 2) the AP and the IRS estimate the AP-user and IRS-user channels, respectively\footnote{Since in practice the AP and IRS are deployed at fixed locations, we assume for simplicity   that  the channel $\G$ between them is quasi-static and changes much slower as compared to the AP-user and IRS-user channels, and thus is constant and known for the considered communication period of interest.}; 3) the AP starts to transmit data to the  users and in the meanwhile  sends its estimated AP-user channels to the IRS controller via a separate control link (see Fig. \ref{system:model}), so that the IRS can jointly optimize the AP transmit beamforming vectors and its phase shifts by using the proposed algorithms in this paper; 4) the IRS controller sends optimized transmit beamforming vectors to the AP and sets its phase shifts accordingly; and 5) the AP and IRS start to transmit data to the users collaboratively. As such, different from  the massive MIMO system, the IRS-aided system generally incurs additional delay  in steps 3) and 4) due to information exchange and algorithm computation. We denote the channel coherence time as $T_c$  and the total delay caused by 3) and 4) as $\tau$, where we assume that $\tau<T_c$. The delay ratio is thus given by $\rho=\frac{\tau}{T_c}$. Since steps 1) and 2) are required in both the IRS-aided system and massive MIMO system, the time overhead for channel training  is omitted for both schemes for a fair comparison.
Note that users are served by the AP only over $\tau$  with the achievable SINR of user $k$, $k\in \mathcal{K}$, denoted by $\text{SINR}^1_k$, while when  they are served by both the AP and IRS during the remaining time of $T_c-\tau$,  the SINR of user $k$ is denoted by $\text{SINR}^2_k$. Accordingly, the achievable rates of user $k$ in the above two phases  can be expressed as $R^1_{k}=\log_2( 1+ \text{SINR}^1_k )$ and $R^2_{k}=\log_2( 1+ \text{SINR}^2_k )$ in bps/Hz, respectively.   The average achievable rate of user $k$ over $T_c$ is thus given by $r_k=\frac{ \tau}{T_c} R^1_{k} + (1-\frac{\tau}{T_c}) R^2_{k}= \rho R^1_{k} + (1-\rho) R^2_{k}$.
Given the transmit power constraint at the AP, we aim to maximize the minimum achievable rate of $r_k$ among all the users, which is a dual problem of (P1) and thus can be solved by using the proposed alternating optimization algorithm together with an efficient bisection search over the AP transmit power  \cite{wiesel2006linear}.  For the case of massive MIMO, we adopt the optimal MMSE-based transmit  beamforming at the AP. }

\begin{figure*}[!t]
\begin{minipage}[t]{0.5\linewidth}
\centering
\includegraphics[width=3.3in, height=2.4in]{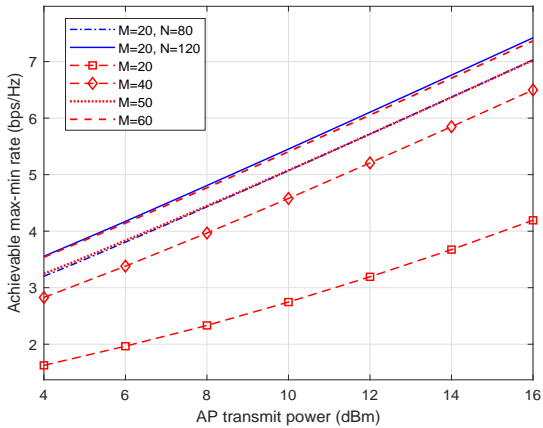}  
\caption{{Achievable max-min rate versus transmit power at AP.}          }
 \label{rate:power}
\end{minipage}%
~~~~\begin{minipage}[t]{0.5\linewidth}
\centering
\includegraphics[width=3.3in, height=2.4in]{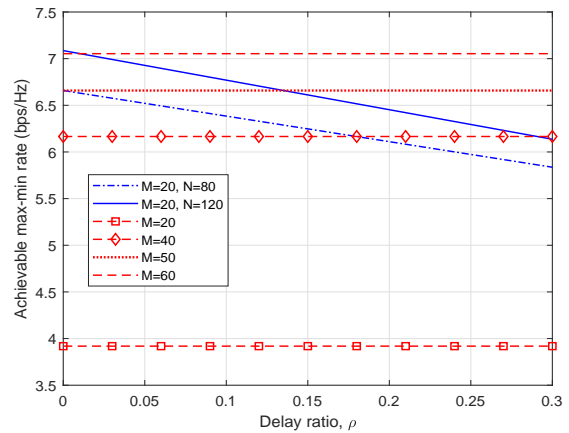}  
\caption{{Impact of IRS delay on  achievable max-min  rate.}      }
\label{rate:delay:tradeoff}
\end{minipage}\vspace{-4mm}
\end{figure*}
{In Fig. \ref{rate:power}, we show the achievable max-min  rate  versus the transmit power at the AP for the ideal case with negligible delay (i.e., $\rho=0$)  in the IRS-aided  system, which provides a throughput  upper bound for practical implementation.
 From Fig. \ref{rate:power}, it is observed that a hybrid deployment of an AP with $M=20$ active antennas and an IRS with $N=80$ passive reflecting elements  achieves nearly the same performance as deploying only an AP with $M=50$ but without using the IRS, which implies that the IRS is indeed effective in reducing the number of active antennas  required in conventional  massive MIMO. To take into account the practical delay due to IRS-AP coordination, we show in Fig. \ref{rate:delay:tradeoff} the achievable max-min rate versus the delay ratio $\rho $ by fixing  the transmit power at the AP as 15 dBm. It is observed  that as the delay ratio increases, the achievable rate of the IRS-aided system decreases since users are served by a relatively small MIMO system (with $M=20$) for more time before the IRS is activated for signal reflection. However, it is also  observed that even with a large  delay of $0.18T_c$, the IRS-aided system with $M=20$ and $N=80$ can still achieve better performance than the AP-only system with $M=40$. This validates the practical throughput gain  of IRS even by taking into account a moderate  delay for its coordination with the AP. }

\section{Conclusions}
In this paper, we proposed a novel  approach to enhance the spectrum and energy efficiency as well as reducing the implementation cost of future wireless communication systems  by leveraging the passive IRS via smartly adjusting its signal reflection. Specifically, given the user  SINR constraints, the active transmit beamforming at the AP and passive reflect beamforming  at the IRS were jointly optimized to minimize the transmit power in an IRS-aided multiuser system. By applying  the SDR and alternating optimization techniques,  efficient algorithms were proposed to trade off between the system performance and computational complexity.
It was shown for the single-user system   that
the receive SNR  increases quadratically with the number of reflecting elements of the IRS, which is more efficient  than the conventional massive MIMO or multi-antenna AF relay.  While for the multiuser system, it was shown that IRS-enabled  interference suppression can be jointly designed with the AP transmit beamforming to improve the performance of all users in the system, even for those that are far away from the IRS.
Extensive  simulation results under various practical setups  demonstrated that  by deploying the IRS and jointly optimizing its reflection with the AP transmission, the wireless network performance can be significantly improved in terms of energy saving, coverage extension as well as achievable rate,   as compared to the conventional setup without using the IRS such as the massive MIMO.  Useful  insights on optimally deploying the IRS and its delay-performance trade-off were also drawn to provide useful guidance for  practical design and implementation.

{In practice, if IRS is equipped with receive RF chains, the commonly used pilot-assisted channel estimation methods can be similarly applied to the IRS as shown in Section V-B; otherwise, it is infeasible for the IRS to directly estimate the channels with its associated  AP/users. For the latter (more challenging)  case,  a viable approach may be to design  the IRS's passive beamforming based on the feedback from the AP/users that receive the signals reflected by the IRS, which is worth investigating in the future work.
In addition, after this paper was submitted, we became aware of another parallel work \cite{huang2018large}, which shows that the IRS-aided wireless system is more energy-efficient than the conventional multi-antenna  AF relay system with HD operation. Although the spectrum efficiency can be further improved by using the FD AF relay, effective SIC is required,  which incurs additional energy consumption. As such, it is worthy of further comparing the energy efficiency of IRS with the FD AF relaying in future work.}
\bibliographystyle{IEEEtran}
\bibliography{IEEEabrv,mybib}

\end{document}